\documentclass{article}
    \usepackage[a4paper]{geometry}
	\usepackage{amsmath, amsthm,dsfont}
	\usepackage{amssymb}
	\usepackage{url}
	\usepackage[dvipsnames]{xcolor}
	\usepackage{bbm}
	\usepackage{tikz}
	\usepackage{thm-restate}

	\usepackage{thmtools}
	\usepackage{thm-restate}

\usepackage{cleveref}

\newtheorem{theorem}{Theorem}
\newtheorem{lemma}[theorem]{Lemma}

\newtheorem*{theorem*}{Theorem}
\newtheorem*{cor*}{Corollary}

\theoremstyle{definition}
\newtheorem{definition}{Definition}
\newtheorem{claim}[theorem]{Claim}

\theoremstyle{remark} 
\newtheorem{remark}[theorem]{Remark}

\newenvironment{claimproof}[1]{\textit{#1.} }{\hfill$\square$}

\newcommand{\N}{\mathbb{N}}
\newcommand{\E}{\mathbb{E}}

\newcommand{\Z}{\mathbb{Z}}

\newcommand{\eps}{\varepsilon}
\renewcommand{\epsilon}{\eps}

\newcommand{\Tcov}{\E(t_{cov})}

\title{Tight Bounds for the Cover Times of Random Walks\\ with Heterogeneous Step Lengths
\thanks{This work has received funding from the European Research Council (ERC) under the European Union's Horizon 2020 research and innovation program (grant  agreement No 648032).}}
\author{
  Brieuc Guinard \\
  \texttt{guinard@irif.fr}
  \and
  Amos Korman \\
  \texttt{amos.korman@irif.fr} \\
\and 
IRIF, CNRS and Univ. of Paris, France}
\date{}

\begin{document}

\maketitle

\begin{abstract}\parskip0.08cm
	Search patterns of randomly oriented steps of different lengths have been observed on all scales of the biological world, ranging from the microscopic to the ecological, including in protein motors, bacteria, T-cells, honeybees, marine predators, and more, see e.g., \cite{MarinePredatorEnvironment, Jansen, WeierstrassianSnails,Amoeba,MarinePredator,Viswanathan,Viswanathan2}. Through different models, it has been demonstrated that adopting a variety in the magnitude of the step lengths can greatly improve the search efficiency. However, the precise connection between the search efficiency and the number of step lengths in the repertoire of the searcher has not been identified. 
	
	Motivated by biological examples in one-dimensional terrains, a recent paper studied the best cover time on an $n$-node cycle that can be achieved by a random walk process that uses $k$ step lengths \cite{LATIN}. By tuning the lengths and corresponding probabilities the authors therein showed that the best cover time is roughly $n^{1+\Theta(1/k)}$. While this bound is useful for large values of $k$, it is hardly informative for small $k$ values, which are of interest in biology \cite{Auger,Review-inter,Lomholt,WeierstrassianMussels}. In this paper, we provide a tight bound for the cover time of such a walk, for every integer $k> 1$. Specifically, up to lower order polylogarithmic factors, the cover time is $n^{1+\frac{1}{2k-1}}$. For $k=2,3, 4$ and $5$ the bound is thus $n^{4/3}$, $n^{6/5}$, $n^{8/7}$, and $n^{10/9}$, respectively. Informally, our result implies that, as long as the number of step lengths $k$ is not too large, incorporating an additional step length to the repertoire of the process enables to improve the cover time by a polynomial factor, but the extent of the improvement gradually decreases with $k$. 
\end{abstract}
	\newpage
	\setcounter{page}{1}
	\section{Introduction}
	This paper follows the ``Natural Algorithms'' line of research, aiming to contribute to biological studies from an algorithmic perspective \cite{ITCS,chazelle1,ANTS, Musco}. In particular, we follow a similar approach to Chazelle \cite{chazelle1,chazelle2}, considering a process that has been extensively studied by physicists and biologists, and offering a more uniform algorithmic analysis based on techniques from probability theory. Our subject of interest is random walks with heterogeneous step lengths, a family of processes that during the last two decades has become a central model for biological movement, see e.g., \cite{COVER,MarinePredatorEnvironment,Jansen,Lomholt,oneD, HunterGatherer,WeierstrassianSnails,humans2,MarinePredator,Viswanathan,Viswanathan2}. Our approach is to quantify by how much can the search efficiency improve when the searcher is allowed to use more steps. Specifically, our goal is to analyze, for every integer $k$, the best cover time achievable by a random walk that utilizes $k$ step-lengths, and identify the parameters that achieve the optimal cover time. Hence, in some sense, we view the number of steps as a ``hardware'' constraint on the searcher, and ask what is the best ``software'' to utilize them, that is, the best way to set the lengths, and the probabilities of taking the corresponding steps. We focus on the one-dimensional terrain (an $n$-node cycle) as it is both biologically relevant, and, among other Euclidean spaces, it is the most sensitive to step-length variations (e.g., the simple random walk on the two-dimensional plain already enjoys a quasi-linear cover time). 
	A preliminary investigation of this question was recently done by the authors of the current paper together with collaborating researchers \cite{LATIN}, yielding asymptotic bounds with respect to $k$. Unfortunately, these bounds are not very informative for small values of $k$, which are of particular interest in biology \cite{Auger,Review-inter,Lomholt,WeierstrassianMussels}. For example, for processes that can use a small number of step-lengths, say $k=2$ or $k=3$, the bound in \cite{LATIN} merely says that the cover time is polynomial in $n$, which does not even imply that such a process can outperform the simple random walk--- whose cover time is known to be $\Theta(n^2)$. In this paper we improve both the lower bound and the upper bound in \cite{LATIN}, identifying the tight cover time for every integer $k$.


		\subsection{Background and Motivation}
	The exploration-exploitation dilemma is fundamental to almost all search or foraging processes in biology \cite{Couzin}. An efficient search strategy needs to strike a proper balance between the need to explore new areas and the need to exploit the more promising ones found. At an intuitive level, this is often perceived as a tradeoff between two scales: the global scale of exploration and the local scale of exploitation. This paper studies the benefits of incorporating a hierarchy of multiple scales, where lower scales serve to exploit the exploration made by higher scales. We demonstrate this concept by focusing on random walk search patterns with heterogeneous step lengths, viewing the usage of steps of a given length as searching on a particular scale.

	In the last two decades, random walks with heterogeneous step lengths have been used by biologists and physicists to model biological processes across scales, from  microscopic to macroscopic, including in DNA binding proteins \cite{Berg,DNA}, immune cells \cite{Tcell}, crawling amoeba \cite{Amoeba}, locomotion mode in mussels \cite{Jager,Jansen}, snails \cite{WeierstrassianSnails}, marine predators \cite{MarinePredatorEnvironment,MarinePredator}, albatrosses \cite{Viswanathan,Viswanathan2}, and even in humans \cite{Boyer1743,humans2,HunterGatherer}. Most of these biological examples concern search contexts, e.g., searching for pathogens or food.	Indeed, from a search efficiency perspective, it has been argued that the heterogeneity of step lengths in such processes allows to reduce oversample, effectively improving the balance between global exploration and local exploitation \cite{Review-inter,Viswanathan2}. 
However, the precise connection between the search efficiency and the number of step lengths in the repertoire of the searcher has not been identified. 

	Due to possible cognitive conflicts between motion and perception, in some of the aforementioned search contexts it was argued that  biological entities 
	are essentially unable to detect targets while moving fast, and hence targets are effectively found only between jumps, see e.g., \cite{Review-inter,Lomholt} and the references therein. Those models are often called {\em intermittent}. When the search is intermittent, we say that a site is {\em visited} whenever the searcher completes a jump landing on this site.  It is also typically assumed that the searcher has some radius of visibility $r$, and a target can only be detected if it is in the $r$-vicinity of a site currently visited by the searcher. Discretizing the space, one may view a Euclidian space as a grid of the appropriate dimension, in which each edge is of length $r$. In this discretization, sites are nodes, and the searcher can detect a target at a node, only if it makes a random jump that lands on it. 
	
	In general, two families of processes with heterogeneous step lengths have been extensively studied in Euclidean spaces: {\em L\'evy Flights} (named after the mathematician Paul L\'evy), and {\em Composite Correlated Random Walks (CCRW)}, see e.g., \cite{Auger,Review-inter,Lomholt}. Both have been claimed to be optimal under certain conditions and both have certain empirical support. In the L\'evy Flight process, step lengths have a probability distribution that is heavy-tailed: at each step a direction is chosen uniformly at random, and the probability to perform a step of length $d$ is proportional to $d^{-\mu}$, for some fixed parameter $1<\mu<3$. 
	 
	Searchers employing a CCRW  can potentially alternate between multiple modes of search\footnote{CCRW have also been classified as either {\em cue-sensitive}, i.e., they can change their mode of operation upon detecting a target \cite{Benhamou}, or {\em internally-driven}, i.e., their movement pattern depends only on the mechanism internal to the searcher \cite{InherentLevy}. 
		However, when targets are extremely rare and there is no a-priori knowledge about their distribution, one must cover a large portion of the terrain before finding a target, and hence the aforementioned distinction becomes irrelevant.}, but apart for few exceptions \cite{WeierstrassianMussels}, such patterns have mostly been studied when assuming that the number of search modes is 2. Specifically, a diffusive phase  in which targets can be detected and a ballistic phase in which the searcher moves in a random direction in a straight line whose length is exponentially distributed with some mean $L$. This CCRW with 2 modes can be approximated as a discrete random walk with two step lengths, hereafter called {\em 2-scales search}: first, choose a direction uniformly at random. Then, with some probability $p$ take a step of unit length, and otherwise, with probability $1-p$, take a step of some predetermined length $L$.

	L\'evy Flights  and  2-Scales searches have been studied extensively using differential equation techniques and  computer simulations. These studies aimed to both compare the performances of these processes as well as  to identify the parameters that maximize the rate of target detection or minimize the hitting time under various target distributions \cite{Review-inter,COVER,Lomholt,oneD,Viswanathan2}.

	Most of the literature on the subject has concentrated on either one or two dimensional Euclidian spaces. In particular, the one-dimensional case has attracted attention due to several reasons. First, it finds relevance in several biological contexts, including in the reaction pathway of DNA binding proteins \cite{Berg,DNA}. One-dimension can also serve as an approximation to general narrow and long topologies, which can be found for example  in blood veins or other organs. Second, from a computational perspective, the one-dimension is the only dimension where the simple random walk has a large cover time, namely, quadratic, whereas in all higher dimensions the cover time is nearly linear.
	This implies that in terms of the cover time, heterogeneous random walks can potentially play a much more significant role in one-dimension than in higher dimensions. 

\subsection{Definitions} 
We model the one-dimension space as an 
$n$-node cycle, termed $C_n$. For an integer $k$, we define the random walks process with $k$ step lengths as follows.
	\begin{definition}[$k$-scales search] \label{sec:model}
	A random walk process $X$ is called a \emph{$k$-scales search} on $C_n$  if there exists	a probability distribution ${\textbf{p}}=(p_i)_{i=0}^{k-1}$, where $\sum_i p_i=1$, and integers $L_0,L_1,\ldots, L_{k-1}$ such that, on each step, $X$ makes a jump $\{0,-L_i,+L_i\}$ with probability respectively $p_i/2, p_i/4, p_i/4$. Overall, with probability $1/2$, the process $X$ stays in place\footnotemark.
			The numbers $(p_i)$ and $(L_i)$ are called the \emph{parameters} of the search process $X$. \footnotetext{This laziness assumption is used for technical reasons, as is common in many other contexts of random walks. Note that  this assumption does not affect the time performance of the process, as we consider it takes time $0$ to stay in place.} The speed is assumed to be a unitary constant, that is, it takes $L$ time to do a step of length $L$.
	\end{definition}
	\noindent 
	Our goal is to show upper and lower bounds on the \emph{cover time} of a $k$-scales search, that is, the expected time to visit every node of the ambient graph $C_n$, where it is assumed that a jump from some point $x$ to $y$ visits only the endpoint $y$, and not any of the intermediate nodes. We denote by $\E(t_{cov}(n,k))$ the smallest cover time achievable by a $k$-scales search over the $n$-node cycle. The parameters $n$ and $k$ are omitted when clear from the context.  
	
	We also define the following $k$-scales search which is often referred to in the mathematical literature as a {\em Weierstrassian random walk} \cite{WeierstrassianMaths1}. In the biology literature, it has been 	used as a model for the movement strategy of snails \cite{WeierstrassianSnails} and mussels \cite{WeierstrassianMussels}. 
	\begin{definition}[Weierstrassian random walk] Let $b\geq 2$ and $k$ be integers such that $b^{k-1}< n \leq b^{k}$. The {\em Weierstrassian random walk} with parameter $b$ is the $k$-scales search defined by: 
			$L_i = b^i$ and $p_i = c_b b^{-i}$, 
			for every $0 \leq i \leq k-1$, with the normalizing constant $c_b=\frac{b-1}{b-b^{1-k}}$.
	\end{definition} 
	Note  that $c_b$ is an increasing function of $b>1$, and so $c_b\geq c_2\geq  1/2$ for $b\geq 2$. Hence, $p_0=c_b\geq  1/2$. Also $p_0=c_b\leq 1$, hence $c_b=\Theta(1)$ is indeed a constant. 

	\subsection{Previous Bounds on the Cover Time of $k$-scales search}
	The work of Lomholt et al.\  \cite{Lomholt} considered intermittent search on the one-dimensional cycle of length $n$, and compared the performances of the best 2-scales search  to the best L\'evy Flight. With the best parameters, they showed that the best 2-scales search can find a target in roughly $n^{4/3}$ expected time, but introducing L\'evy distributed relocations with exponent $\mu$ close to $2$ can reduce the search time to quasi linear. 
	
	Taking a more unified computational approach, a recent paper \cite{LATIN} analyzed the impact of having $k$ heterogeneous step lengths on the cover time (or hitting time\footnote{Note that in connected graphs, the notion of {\em cover time}, namely the expected time until all sites (of a finite domain) are visited when starting the search from the worst case site, is highly related to the {\em hitting time}, namely, the expected time to visit a node $x$ starting from node $y$, taken on the worst case pair $x$ and $y$; the cover time is always at least the hitting time, and in connected graphs it is at most  a logarithmic multiplicative factor more than the hitting time, see \cite{Peres}[Matthews method, Theorem~11.2].}) of the $n$-node cycle $C_n$. Specifically, the following bounds were established in \cite{LATIN}.
	\begin{theorem*}[Upper bound on the cover time of Weierstrassian random walk from \cite{LATIN}]\label{thm:main-up-latin}
		Let $b, n$ be integers such that $2\leq b<n$ and set $k=\log n / \log b$. The cover time of the Weierstrassian random walk with parameter $b$ on the $n$-cycle is at most
		$\mbox{poly}(k) \cdot \mbox{poly}(b) \cdot n \log n$.
	\end{theorem*}
	Taking $b=\lceil n^{1/k} \rceil$ yields the following corollary.
	\begin{cor*}[Upper bound from \cite{LATIN}]\label{cor:main-up-latin}
		For any $k\leq \frac{\log n}{\log \log n}$, there exists a $k$-scales search with cover time $n^{1+O(\frac 1k)} \log n$.
		
	\end{cor*}
Note that for small values of $k$, this bound is not very informative. For example, for $k=2,3$ the bound merely says that the cover time is polynomial in $n$, which is known already for $k=1$, i.e., the simple random walk, whose cover time is $\Theta(n^2)$.

	\begin{theorem*}[Lower bound from \cite{LATIN}]\label{thm:main-lb-latin}
		For every $\epsilon>0$, there exist sufficiently small  constants $c,c'>0$ such that for $k\leq c' \frac{\log n}{\log \log n}$, any $k$-scales search cannot achieve a cover time better than $c\cdot n^{1+\frac{1/2-\epsilon}{k+1}}$. 
	\end{theorem*}
		The aforementioned lower bound of \cite{LATIN} is more precise than the upper bound, but still not tight, as we show in the next subsection. For example, for $k=2$, the lower bound in \cite{LATIN} gives $n^{7/6}$ instead of $n^{4/3}$, which is the tight bound. 
	
	\subsection{Our Results}
 This paper provides tight bounds for the cover times of $k$-scales searches, for any integer $k>1$. Specifically, we prove that the optimal cover (or hitting) time achievable by a $k$-scales search is $n^{1+\frac{1}{2k-1}}$, up to lower order polylogarithmic factors. Our bound implies that for small $k$, the improvement in the cover time incurred by employing one more step length is polynomial, but the extent of the improvement gradually decreases with $k$. 
 
 In order to establish the tight bound, we first had to understand what should be a good candidate for the tight bound to aim to. This was not a trivial task, as the precise bound takes an unusual form. After identifying the candidate for the bound, we had to improve both the upper and the lower bounds from \cite{LATIN}, which required us to overcome some key technical difficulties. For the lower bound, \cite{LATIN} established that the cover time is bounded from below by a function (specifically the square-root) of the ratio $L_{i+1}/L_i$, for every $i$. As it turns out, what was required to tighten the analysis is a better understanding about the relationships between the cover time and the extreme step-lengths, namely, $L_0$, $L_1$ and $L_{k-1}$. Specifically, in proving the precise lower bound we have two components, one for the ``local'' part (exploitation) and the other for the ``global'' part (exploration). We showed that in order to be efficient on the local part, the small step-lengths need to be small, whereas in order to be  efficient on the global part (traversing large distances fast), the largest step-length needs to be large. This allowed us to widen the ratios between consecutive step-lengths, consequently increasing the lower bound. 
 
 In order to obtain the precise upper bound, we improved the analysis in \cite{LATIN} of the Weierstrassian random walk process. This, in particular, required overcoming non-trivial issues concerning dependencies between variables that were overlooked in \cite{LATIN}. By doing this, we also refined the estimates on the order of magnitude of other dependencies. In addition,  we had to incorporate short-time probability bounds for each step-length used by the process, and perform a tighter analysis of the part of the walk that corresponds to the largest step length $L_{k-1}$. 
 
 We next describe our contribution in more details.

	\subsubsection{The Lower Bound}
	We begin with the statement of the lower bound. The formal proof is given in Section \ref{sec:lower}.
	\begin{theorem}\label{thm:main-lb}
		Let $k$ and $n$ be positive integers. The cover time of any $k$-scales search $X$ on $C_n$ is: \[\E(t_{cov}(n,k))=n^{1+\frac{1}{2k-1}}\cdot \Omega(1/k).\]
	\end{theorem}
	
	\subsubsection{The Upper Bound}
	The following theorem implies that up to lower order terms, the cover time of the Weierstrassian random walk  matches the lower bound  of the cover time of any $k$-scales search, as given by Theorem~\ref{thm:main-lb}, for $2\leq k\leq \log n$, i.e., for all potential scales. 
	\begin{theorem}	\label{thm:main-upper}	Let $k$ be an integer such that $2\leq k\leq \log_2 n$. The  Weierstrassian random walk with parameter $b=\lfloor n^{\frac{2}{2k-1}} \rfloor $ is a $k$-scales search that achieves a cover time of: \[n^{1+\frac{1}{2k-1}} \cdot O\left(k^2\log^{2} n\right).\] 
	\end{theorem} 
	Observe that combining Theorems~\ref{thm:main-lb} and \ref{thm:main-upper} we obtain the best cover time $Cov_{k,n}$ achievable by a $k$-scales search on $C_n$, which is $\tilde{\Theta}\left(n^{1+\frac{1}{2k-1}}\right)$ for any $2\leq k\leq \log n$. For particular values of $k$, we thus have:
	\begin{equation*}
	\begin{array}{|c|c|c|c|c|c|c|c|}
	\hline
	k & 1 & 2 & 3 & 4 & 5 & \ldots & \log n
	\\
	\hline 	
	\E(t_{cov}(n,k)) & {\Theta}(n^2) & \tilde{\Theta}(n^{\frac 43}) & \tilde{\Theta}(n^{\frac 65}) & \tilde{\Theta}(n^{\frac 87}) &  \tilde{\Theta}(n^{\frac{10}{9}})&\ldots &O(n \log^3 n)  \\
	\hline
	\end{array}
	\end{equation*}
	
	Theorem \ref{thm:main-upper} follows immediately from the following more general theorem, by taking $b=n^{\frac{2}{2k-1}}$.
	
		\begin{restatable}[]{theorem}{ThmUp}\label{thm:main-up}
		Let $b, k, n$ be integers such that $b^{k-1}<n\leq b^k$.
		The cover time of the Weierstrassian random walk on $C_n$ with parameter $b$ is
		\[ O\left( n \max\left\{\frac{b^k}{n},\frac{n}{b^{k-1}}\right\} \cdot k^2 \cdot \log b\cdot \log n\right )=\tilde{O}\left(\max\left\{{b^k},\frac{n^2}{b^{k-1}}\right\} \right). \]

\end{restatable}

The formal proof of Theorem \ref{thm:main-up} is deferred to Appendix \ref{sec:upper-bound-proof}. In Section \ref{sec:upper-intuition} we provide a sketch of the proof.

As mentioned, Theorem \ref{thm:main-up} using the particular value $b=n^{\frac{2}{2k-1}}$ gives a tight upper bound for $k$-scales search. However, since the Weierstrassian random walk is of independent interest as it is used in biology, it might be useful to understand
its cover time also for other values of $b$. 
Note that Lemmas \ref{lem:main-lb} and \ref{lem:lb-bigsteps} below, when applied to the Weierstrassian random walk on $C_n$, show that the cover time is at least
$\Omega\left ( \max\{ n\sqrt{b},\frac{n^2}{b^{k-1}} \}\right ).$
This is quite close to the bound  $\tilde{O}\left(\max\left\{{b^k},\frac{n^2}{b^{k-1}}\right\} \right)$ of the theorem. Indeed if $n\geq b^{k-\frac 12}$, both bounds match, up to logarithmic terms. If $n\leq b^{k-\frac 12}$, the ratio of the bounds is $\frac{b^{k-\frac 12}}{n}$.

	\section{The Lower Bound Proof}\label{sec:lower}
	The goal of this section is to establish the lower bound in Theorem \ref{thm:main-lb}.
	For this purpose, 
	consider a $k$-scales search $X$ on the cycle $C_n$ and denote $(L_i)_{i=0}^{k-1}$ its step lengths with $L_i  < L_{i+1}$ for all $i\in [k-2]$. For convenience of writing we also set $L_{k} = n$, but it should be clear that it is actually not a step length of the walk. Let $p_i$ denote the probability of taking the step length $L_i$.
	
	The theorem will follow from the combination of two lemmas. The first one, Lemma \ref{lem:main-lb}, stems from the analysis of the number of nodes that can be visited during $L_{i+1}$ time steps. It forces $L_0L_1$ as well as the ratios $L_{i+1}/L_i$ for all $1\leq i\leq k-1$ to be small enough in order to have a small cover time. The second one, Lemma~\ref{lem:lb-bigsteps}, comes from bounding the cover time by the time it takes to go to a distance of at least $n/3$. It forces $L_{k-1}$ to be big enough to have a small cover time.
	
	\begin{lemma}\label{lem:main-lb}
		The cover time of $X$ is  at least 
		\begin{itemize}
			\item $\Tcov=\Omega(n\sqrt{L_0 L_1})$.
			\item  $\Tcov=\Omega\left(\frac{n}{k}  \sqrt{\frac{L_{i+1}}{L_i}}\right)$ for any $1\leq i\leq k-1$.
		\end{itemize}
	\end{lemma}
	
	The second part of Lemma \ref{lem:main-lb} was already given in \cite{LATIN}. We sketch here the ideas behind the proof of the first part, namely, that the cover time is at least of order $n\sqrt{L_0 L_1}$. Essentially, we count the expected number of nodes that can be visited in a time duration of $L_1$, which we call a {\em phase}. A jump of length $L_i\geq L_1$ will not contribute to visiting a new node during this time duration. Thus, we may suppose that there are only jumps of length $L_0$. Since $L_1 \leq n$, the process does not do a turn of the cycle and, therefore, it can be viewed as a walk on $\Z$. Furthermore, since every jump has length $L_0$, we can couple this walk by a corresponding simple random walk, that does steps of length $1$, during a time duration of ${L_1}/{L_0}$. The expected number of nodes visited during a phase is thus of order $\sqrt{{L_1}/{L_0}}$. It follows that we need at least $n /(\sqrt{{L_1}/{L_0}})$ such phases before covering the cycle. Since a phase lasts for $L_1$ time, the cover time is at least of order $n \sqrt{L_0L_1}$. The full proof of Lemma~\ref{lem:main-lb}, including the part that was proven in \cite{LATIN}, appears, for completeness, in Appendix \ref{app:lem1}.

	\begin{lemma} \label{lem:lb-bigsteps}The cover time of $X$ is at least $\Omega(n^2 \frac{\mu}{\sigma^2})$, 
		where $\mu=\frac 12 \sum_{i\leq k-1}p_iL_i$ and $\sigma^2=\frac 12 \sum_{i\leq k-1}p_iL_i^2$ are the mean and variance of the jump lengths, respectively. In particular, the cover time is: \[\Tcov=\Omega\left(\frac{n^2}{L_{k-1}}\right).\]
	\end{lemma}
	
	\begin{proof}
	Let $m_{cov}$ denote the random number of \textit{steps} before all nodes of $C_n$ are covered, and let $t_{cov}$ be the random cover time of the process. By Wald's identity, we have:
	\begin{equation}\label{eq:time-from-move} \E(t_{cov})=\E(m_{cov})\cdot \mu,
	\end{equation}
	where $\mu=\frac 12 \sum_{i=0}^{k-1} p_i L_i $ is the expected length, and hence the expected time, of a jump (the factor $\frac 12$ comes from the laziness). By Markov's inequality, we have: \[\Pr \left( m_{cov} < 2 \E(m_{cov})\right) \geq  1/2.\]
	Let $N_m$ be the (random) number of  nodes visited by step $m$. We have:
	\[\E(N_{2\E(m_{cov})}) \geq \E\left(N_{2\E(m_{cov})} \mid m_{cov} < 2 \E(m_{cov})\right)\cdot \Pr\left (m_{cov} < 2 \E(m_{cov})\right) \geq n\cdot  \frac 12. \]
	Define $D_m$ as the maximal distance of the process from step $0$ up to step $m$, i.e., $D_m=\max_{s\leq m} \lvert X(s)\rvert $. Since $N_m\leq 2D_m+1$, we have:
	\[ 2\E(D_{2\E(m_{cov})})+1\geq \E(N_{2\E(m_{cov})}) \geq  n/2. \]
	As shown in \cite{Majumdar} for general one-dimensional random walks, we have $\E(D_m)=O\left( \sigma \sqrt{m} \right )$, 
	where $\sigma$ is the standard deviation of the length distribution, i.e., $\sigma^2=\frac 12 \sum_i p_i L_i^2$. Thus, we have:
	\[ \sqrt{\E(m_{cov})}\sigma =\Omega(n), \]
	and so:
	\[ \E(m_{cov}) =\Omega\left( \frac{n^2}{\sigma^2} \right) , \]
	and by Eq.~\eqref{eq:time-from-move}, we get:
	\[ \E(t_{cov}) =\Omega\left( n^2 \frac{\mu}{\sigma^2} \right) = \Omega \left( n^2 \frac{\sum_{i=0}^{k-1} L_ip_i}{\sum_{i=0}^{k-1} L_i^2 p_i}\right) , \]
	which proves the first part of the lemma. 
	
	In order to prove the second part, note that since $L_{k-1}$ is the biggest step length, we have  $\sum_{i=0}^{k-1} p_iL_i (1-\frac{L_i}{L_{k-1}})\geq 0$, and so $\frac{\sum_{i=0}^{k-1}L_ip_i}{\sum_{i=0}^{k-1} L_i^2 p_i} \geq \frac{1}{L_{k-1}}$. Therefore,
	\[ \E(t_{cov}) =  \Omega\left ( \frac{n^2}{L_{k-1}}\right) , \]
	which completes the proof of Lemma \ref{lem:lb-bigsteps}. \end{proof}
	
	Next, it remains to show how Theorem \ref{thm:main-lb} follows by combining Lemma \ref{lem:main-lb} and Lemma \ref{lem:lb-bigsteps}. First, consider the lower bound of $\Omega({n^2}/{L_{k-1}})$ in Lemma \ref{lem:lb-bigsteps}. If $L_{k-1}\leq n^{1-\frac{1}{2k-1}}$ then the bound in Theorem \ref{thm:main-lb} immediately follows. Let us therefore assume that $L_{k-1}>n^{1-\frac{1}{2k-1}}$. 
	
	Define $\alpha_0=L_0L_1$ and $\alpha_i=\frac{L_{i+1}}{L_i}$ for $i\in\{1,2, \ldots,k-2\}$. As
	
	\[ \prod_{i=0}^{k-2} \alpha_i = L_0L_{k-1}, \]
	there must exists an index $0\leq i\leq k-2$ such that $\alpha_i \geq \left(L_0L_{k-1}\right)^{\frac{1}{k-1}}$. Thus, by Lemma \ref{lem:main-lb}, the cover time is at least 
	\[ \Omega\left(\frac{n}{k} \left(L_0L_{k-1}\right)^{\frac{1}{2(k-1)}}\right). \]
	Since $L_{k-1}>n^{1-\frac{1}{2k-1}}=n^{\frac{2k-2}{2k-1}}$ and $L_0\geq 1$, we conclude that the cover time is at least
	
	\[ \Tcov=\Omega\left(\frac{n}{k} \cdot n^{\frac{1}{2k-1}}\right),\] 
	as desired. This completes the proof of Theorem \ref{thm:main-lb}.\qed

	\section{Upper Bound Proof (Sketch)}\label{sec:upper-intuition}

	Let us give the key ideas of the proof of Theorem \ref{thm:main-up}. 
	Some of the initial steps in the proof follow the technique in \cite{LATIN} (by doing so, we also corrected some mistakes in \cite{LATIN}). These parts are clearly mentioned below. Our main technical contribution that allowed us to obtain the precise upper bound, is the use of short-time probability bounds (see Eq.~\eqref{eq:R-distribution}), and a tighter analysis of the part of the walk that corresponds to the largest step length $L_{k-1}$. 
	
	In more details, let us consider the Weierstrassian walk on $C_n$, termed $X$.
	The following lemma establishes a link between the cover time of $X$ and the point-wise probabilities of $X$. For completeness, we provide a formal proof of it in Appendix \ref{sec:upper-bound-proof}, although it is not hard to obtain it using the technique in \cite{LATIN}. 
	
	\begin{restatable}[]{lemma}{PreliminaryLemma}\label{lem:cov-from-pointwise}
		If $p>0$ and $m_0>0$ are such that, for any $x\in \{0,\dots, n-1\}$,
		\begin{align}\label{eq:trick-refined}
		\frac{\sum_{m=m_0}^{2m_0}  \Pr(X(m)=x)}{\sum_{m=0}^{m_0} \Pr(X(m)=0)} \geq p,
		\end{align}
		then the cover time of the Weierstrassian random walk $X$ on the cycle $C_n$ is $O\left (m_0 p^{-1} k \log n \right )$.	
	\end{restatable}
	Using Lemma \ref{lem:cov-from-pointwise}, the bound of Theorem \ref{thm:main-up} can be established by proving bounds on the probability to visit node $x\in [0,n)$ at step $m$.
	
	In order to simplify the presentation, assume first that  $n=b^k$. Proceeding first as in \cite{LATIN}, we view the $k$-lengths Weierstrassian random walks as $k$ (dependent) random walks, by grouping together the jumps of the same length (see Figure \ref{fig:intuition-ub}). Define $S_i(m)$ as the algebraic count of the jumps of lengths $b^i$. E.g., if, by step $m$, there are exactly four positive jumps of length $b^i$, and one negative, then $S_i(m)=3$. We have: \[X(m)=\sum_{i=0}^{k-1} S_i(m) b^i.\] Define also the following decomposition of $C_n$.
	
	\begin{definition}[Base $b$ decomposition]
		For any $x \in C_n$, we may decompose $x$ in base $b$ as \[x = \sum_{i=0}^{k-1} x_i b^i,\] with $0\leq x_i< b$. We call $x_i$ the $i$-th coordinate of $x$ (in base $b$). 
	\end{definition}
	
	It follows from Euclidean division, and the fact that $n=b^k$, that the base $b$ decomposition is well-defined and unique for every $x \in C_n$. This decomposition is illustrated in Figure \ref{fig:intuition-ub} (where we have taken $n=\hat{n}b^{k-1}$ to anticipate the more general case to follow).
	
	Note that $X(m)=x$ in $C_n$ if and only if \begin{equation}
	\label{eq:S-mod-b}
	\sum_i (S_i(m)-x_i)b^i = 0\mod n.
	\end{equation} By taking Eq.~\eqref{eq:S-mod-b} modulo $b^i$, for $i\leq k-1$, it is easy to show that Eq.~\eqref{eq:S-mod-b} is equivalent to
	\begin{equation*}
	S_i(m) =y_i\mod b,
	\end{equation*}
	for $y_i:=x_i-b^{-i}\sum_{j<i} (S_j(m)-x_j)b^j \mod b$.
	
	Thus, $X(m)=x$ is equivalent to $R_i(m)=y_i$ for all $i$, where $R_i=S_i\mod b$ is a random walk on $C_b$ that moves with probability $\frac{p_i}{2}$. This process is illustrated in Figure \ref{fig:intuition-ub}, where $X(m)=7$ is equivalent to $R_0(m)=3$ and $R_1(m)=2$.
	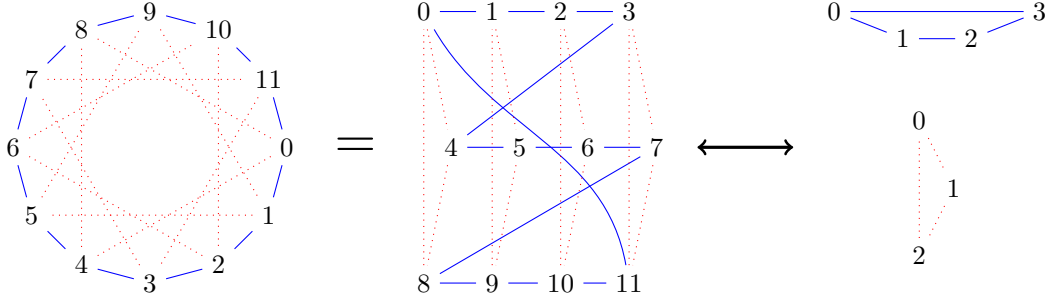
\begin{figure}
		\centering
		\begin{tikzpicture}[scale=1.8]
		
		\tikzset{NodeTe/.style={scale=1}};
		
		\path (0:1) node [NodeTe] (0) {$0$};
		\path (30:1) node [NodeTe] (11) {$11$};
		\path (60:1) node [NodeTe] (10) {$10$};
		\path (90:1) node [NodeTe] (9) {$9$};
		\path (120:1) node [NodeTe] (8) {$8$};
		\path (150:1) node [NodeTe] (7) {$7$};
		\path (180:1) node [NodeTe] (6) {$6$};
		\path (210:1) node [NodeTe] (5) {$5$};
		\path (240:1) node [NodeTe] (4) {$4$};
		\path (270:1) node [NodeTe] (3) {$3$};
		\path (300:1) node [NodeTe] (2) {$2$};
		\path (330:1) node [NodeTe] (1) {$1$};

		\draw [blue] (0)--(1)--(2)--(3)--(4)--(5)--(6)--(7)--(8)--(9)--(10)--(11)--(0);
		
		\draw [red, dotted] (0)--(4)--(8)--(0);
		\draw [red, dotted] (1)--(5)--(9)--(1);
		\draw [red, dotted] (2)--(6)--(10)--(2);
		\draw [red, dotted] (3)--(7)--(11)--(3);
		
		
		\draw(1.5,0) node[scale=2] {$=$};
		
		\path (3.5,-1) node [NodeTe] (11') {$11$};
		\path (3,-1) node [NodeTe] (10') {$10$};
		\path (2.5,-1) node [NodeTe] (9') {$9$};
		\path (2,-1) node [NodeTe] (8') {$8$};
		\path (3.7,0) node [NodeTe] (7') {$7$};
		\path (3.2,0) node [NodeTe] (6') {$6$};
		\path (2.7,0) node [NodeTe] (5') {$5$};
		\path (2.2,0) node [NodeTe] (4') {$4$};
		\path (3.5,1) node [NodeTe] (3') {$3$};
		\path (3,1) node [NodeTe] (2') {$2$};
		\path (2.5,1) node [NodeTe] (1') {$1$};
		\path (2,1) node [NodeTe] (0') {$0$};
		
		\draw [blue] (0')--(1')--(2')--(3');
		\draw[blue] (3')--(4');
		\draw [blue] (4')--(5')--(6')--(7');
		\draw[blue,] (7')--(8');
		\draw[blue] (8')--(9')--(10')--(11');
		\draw [blue, >=latex] (11') to[out=100, in=300] (0');
		
		\draw [red, dotted] (0')--(4')--(8')--(0');
		\draw [red, dotted] (1')--(5')--(9')--(1');
		\draw [red, dotted] (2')--(6')--(10')--(2');
		\draw [red, dotted] (3')--(7')--(11')--(3');

		
		\path (5,1) node [NodeTe] (x00) {$0$};
		\path (5.5,0.8) node [NodeTe] (x01) {$1$};
		\path (6,0.8) node [NodeTe] (x02) {$2$};
		\path (6.5,1) node [NodeTe] (x03) {$3$};
		\draw[blue] (x00)--(x01)--(x02)--(x03)--(x00);
		
		\path (5.62,0.2) node [NodeTe] (x10) {$0$};
		\path (5.87,-0.3) node [NodeTe] (x11) {$1$};
		\path (5.62,-0.8) node [NodeTe] (x12) {$2$};
		\draw[red, dotted] (x10)--(x11)--(x12)--(x10);
		
		\draw[very thick, <->] (4,0) -- (4.7,0);
		\end{tikzpicture}
		\caption{The first two graphs represent, in different node disposition, the Weierstrassian walk on $C_{12}$ with parameter $b=4$. There are $k=2$ jump lengths, $L_0=1$ (blue edges) and $L_1=b=4$ (red, dotted edges). To the right, we show the decomposition of $C_{12}$  as $C_4 \times C_3$. For instance the node $x=7 \in C_{12}$ will be represented by $x_0=3\in C_4$ and $x_1=1\in C_3$.
		}\label{fig:intuition-ub}
	\end{figure}

	Unfortunately, the $R_i$'s and the $y_i$'s are not independent, due to the fact that only one of the $R_i$ can change between steps $m$ and $m+1$, however, let us overlook this issue in this informal outline. We then have:
	\begin{equation}\label{eq:Z-prod-R}\Pr(X(m)=x) \approx \prod_{i=0}^{k-1}\Pr(R_i(m)=y_i). \end{equation}
	Recall that $R_i$ is a random walk over $C_b$ 
	that moves with probability $p_i$. The following is a well-known property of the random walk a cycle (see, e.g., Example 5.7 and Proposition 6.18 in \cite{Aldous}):
	\begin{claim}\label{claim:lwr-distribution} For a simple random walk $R$ on $C_b$ that moves with probability $\frac 12$, and any $y\in C_b$,
		\begin{equation}
		\Pr\left(R(m)=y\right ) = \begin{cases} O\left(1/\sqrt{m}\right) \text{ if } m<b^2 \\  b^{-1}(1\pm \eps_m) \text{ if }m\geq b^2, \end{cases} \end{equation}
		with $\eps_m =O(e^{-cmb^{-2}})$ where $c>0$.
	\end{claim}
 Considering that $R_i$ moves with probability $\frac{p_i}{2}=\Theta(b^{-i})$, we can expect that, at step m, $R_i(m)$ has the same distribution as the lazy random walk with $mp_i$ steps that moves with probability $\frac 12$. This is proved formally in Appendix \ref{sec:upper-bound-proof}. Hence, by substituting $m$ with $mp_i$ in Claim \ref{claim:lwr-distribution}, we obtain:
	\begin{equation}\label{eq:R-distribution} \Pr\left(R_i(m)=y_i\right ) = \begin{cases} O\left(1/\sqrt{mp_i}\right) \text{ if } m<b^{i+2} \\  b^{-1}(1 \pm \eps_{mp_i} ) \text{ if }m\geq b^{i+2}. \end{cases} \end{equation}	
	Theorem \ref{thm:main-up} then follows from Eq.~\eqref{eq:Z-prod-R}, Eq.~\eqref{eq:R-distribution} and Lemma \ref{lem:cov-from-pointwise}. Essentially, to cover $C_n$, we need that each $R_i(m)$ is mixed, i.e., has some significant probability to visit any node $y_i$ in $C_b$, which happens, as shown by Eq.~\eqref{eq:R-distribution}, for $m>b^{k-1+2}=b^{k+1}$. Let us apply Lemma \ref{lem:cov-from-pointwise} with
	\[m_0:= b^{k+1}.\]
	We first establish a lower bound on $\sum_{m=m_0}^{2m_0}\Pr(X(m)=x)$. By Eq.~\eqref{eq:Z-prod-R} and Eq.~\eqref{eq:R-distribution}, we have, for $m>m_0$,
	
	\[\Pr (X(m)=x)\approx \prod_{0\leq i\leq k-1} b^{-1}\left(  1-\eps_{mp_i}\right) = \Theta \left (b^{-k}\right),\]
    where the last equality is justified in the appendix. Thus,
    \[  \sum_{m=m_0}^{2m_0} \Pr(X(m)=x) = \Omega\left ( m_0 b^{-k}\right  ) =\Omega \left( b\right ).\]
	We need also to upper bound $\sum_{m=0}^{m_0} \Pr(X(m)=0)$, which is the expected number of returns to the origin up to step $m_0$.
 To do this, we shall use the short-time bounds of Eq.~\eqref{eq:R-distribution}. 
 
 Let us decompose the aforementioned sum as follows.
 \begin{equation}\label{eq:returns-origin-sum-decomposition} \sum_{m=0}^{m_0} \Pr(X(m)=0) = 1+\frac 12+ \sum_{j=0}^{k-1} \sum_{m=1+b^j}^{b^{j+1}} \Pr(X(m)=0) + \sum_{m=1+b^k}^{m_0} \Pr(X(m)=0).\end{equation}
	Fix $j$, such that $1\leq j \leq k-1$ and let $m\in (b^j, b^{j+1}]$. 
	By Eq.~\eqref{eq:Z-prod-R}, in order to upper bound $\Pr(X(m)=0)$ it is enough to bound $\Pr(R_i(m)=y_i)$ for every $i\leq k-1$. For $i>j$,
	we bound $\Pr(R_i(m)=y_i)$ by $1$. For $i\leq j-2$, we use Eq.~\eqref{eq:R-distribution} to upper bound $\Pr(R_i(m)=y_i)$ by $b^{-1}(1+\eps_{mp_j})$. For $i=j-1$ and $i=j$, we bound $\Pr(R_i(m)=y_i)$ by $O\left(1/\sqrt{mp_{j-1}}\right)$ and $O\left(1/\sqrt{mp_j}\right)$, respectively. We thus obtain, by Eq.~\eqref{eq:Z-prod-R}, 
	
	\begin{align*} \Pr(X(m)=x) &=O\left (\frac{1}{\sqrt{mp_{j-1}}} \cdot \frac{1}{\sqrt{mp_j}}  \cdot \prod_{0\leq i\leq j-2} b^{-1} \left  (1 + \eps_{mp_j} \right ) \right)\\ & = O\left (b^{-(j-1)}  \cdot \frac{\sqrt{b}b^{j-1}}{m} \right )=O\left (  \frac{\sqrt{b}}{m} \right), \end{align*}
where we justify in the appendix that $\prod_{0\leq i\leq j-2} ( 1 + \eps_{mp_j} )=O(1)$. Hence, we get:
\begin{equation}\label{eq:returns-origin-sum-i}
     \sum_{m=1+b^j}^{b^{j+1}} \Pr(X(m)=0) = O( \sqrt{b}\log b),
\end{equation}
by using that $\sum_{m=1+b^j}^{b^{j+1}} m^{-1} = \Theta\left( \int_{m=b^j}^{b^{j+1}} u^{-1}du\right)=\Theta(\log b)$. For the case $j=0$, we bound $\Pr(R_i(m)=y_i)$ by $1$ for $i>1$ and $\Pr(R_0(m)=y_0)$ by $O(m^{-\frac 12})$, so that, by Eq.~\eqref{eq:Z-prod-R}, $\Pr(X(m)=0)=O(\frac{1}{\sqrt{m}})$. Hence, we get:
\begin{equation}\label{eq:returns-origin-sum-0}
 \sum_{m=2}^{b} \Pr(X(m)=0) = O\left(\sqrt{b}\right).\end{equation}  Similarly, for $m\in (b^k, b^{k+1}]$, $\Pr(R_i(m)=y_i)$ is bounded by $b^{-1}(1+\eps_{mp_i})$ for $i\leq k-2$, and by $\frac{1}{\sqrt{mp_{k-1}}}$ for $i=k-1$. Thus, for $m\in (b^k, b^{k+1}]$,  \[\Pr(X(m)=0)=O\left(\frac{1}{\sqrt{m}\sqrt{b^{k-1}}}\right)\] and, since $\sum_{m=1+b^k}^{b^{k+1}} \frac{1}{\sqrt{m}} = O\left(\int_{b^k}^{b^{k+1}} \frac{1}{\sqrt{u}} du\right) = O\left ( \sqrt{b^{k+1}} \right)$, we get:
 \begin{equation}\label{eq:returns-origin-sum-k}
\sum_{m=1+b^k}^{b^{k+1}} \Pr(X(m)=x) = O\left(\frac{\sqrt{b^{k+1}}}{\sqrt{b^{k-1}}}\right)=O(b).\end{equation} In total, by Eq.~\eqref{eq:returns-origin-sum-decomposition}, combining Eqs.~\eqref{eq:returns-origin-sum-i}, \eqref{eq:returns-origin-sum-0} and \eqref{eq:returns-origin-sum-k}, we find that the expected number of returns to the origin up to step $b^{k+1}$ is \[ \sum_{m=0}^{m_0} \Pr(X(m)=0) = O\left ( k\sqrt{b}\log b+b \right ) =O\left ( kb \log b\right ).\]
So that all together we have: \[ \frac{\sum_{m=m_0}^{2m_0}\Pr(X(m)=x)}{\sum_{m=0}^{m_0} \Pr(X(m)=0)} = \Omega \left( \frac{b}{kb\log b} \right) = \Omega \left(\frac{1}{k\log b} \right). \] Thus, by Lemma \ref{lem:cov-from-pointwise}, the cover time of $X$ is at most: \begin{equation}\label{eq:cover-time-bound-main-text-b-k}
    O(m_0\cdot k \log b \cdot k \log n) =O(b^{k+1} k^2 \log b \log n)=O(nbk^2\log b \log n),
\end{equation} as claimed by Theorem~\ref{thm:main-up}, for the case where $n=b^k$.
	
	Consider now a more general case, in which $n$ is a multiple of $b^{k-1}$. Here, we can write $n=\hat{n}b^{k-1}$, where $\hat{n}\in(0,b]$ is an integer. What changes in this case is that the last coordinate, $R_{k-1}$, is now a random walk over $C_{\hat{n}}$ instead of over $C_b$, as depicted in Figure \ref{fig:intuition-ub}. $R_{k-1}$ is thus mixed after number of steps: \[ \hat{n}^2p_{k-1}^{-1}=\Theta(b^{k-1}\hat{n}^2)=\Theta(n^2/b^{k-1}).
	    \]  On the other hand, after $\Theta(b^{k-2+2})=\Theta(b^k)$ steps, the other coordinates are mixed. Thus, the number of steps needed before every coordinate $R_i$ is mixed is: 
	    \begin{equation}\label{eq:multiple} m_0=\Theta\left(\max\{b^k,n^2/b^{k-1}\}\right),\end{equation}
	which is again the order of magnitude of the cover time of $X$, up to polylogarithmic factors. Note that when $n=b^k$, Eq.~\eqref{eq:multiple} recovers the cover time of order $\tilde{\Theta}(b^{k+1})$. Furthermore, the ratio of the cover time for $n=b^{k}$ and $n=\hat{n}b^{k-1}$ is of order $\frac{b^{k+1}}{\max\{b^k,b^{k-1}\hat{n}^2\}}=\min\{b,\frac{b^2}{\hat{n}^2}\}$. When $b$ is large (which corresponds to $k$ being small), this can be significant. 
	Hence, naively bounding $\hat{n}$ from above by $b$ would not suffice to yield an optimal bound.
	
	The general case, when $n$ is not necessarily a multiple of $b^{k-1}$, needs to be treated with more care. What changes in this case is that we can no longer decompose $X$ as $k$ dependent random walks on $C_b \times \dots \times C_b \times C_{\frac{n}{b^{k-1}}} $, since $\frac{n}{b^{k-1}}$ is not an integer. Instead, we define $Z$ as the process that does the same jumps as $X$, but on the infinite line $\Z$, and we also define \[\hat{n}:=\lfloor {n}/{b^{k-1}} \rfloor.\]
	Then, we use almost the same decomposition, where $Z$ is viewed as $k$ dependent random walks over $C_b\times \dots\times C_b \times \Z$. The process corresponding to the last coordinate, $R_{k-1}$, is now a random walk on $\Z$, and we are interested especially on the probability of the event $R_{k-1}(m)=x_{k-1}$ for $x_{k-1}\in [0, \hat{n}]$. As the coordinate $R_{k-1}$ is not restricted to $[0, \hat{n}]$, we need to pay attention that the walk does not go too far. 
	
	\section{Discussion}

	The upper bound in Theorem  \ref{thm:main-upper} implies that almost linear time performances, as those obtained by L\'evy Flights, can be achieved with a number of step lengths that ranges from logarithmic to linear.  This further suggests that cover time performances similar to those of L\'evy Flights can be seen by a large number of different processes. In practice, if one aims to fit  empirical statistics of an observed process to a theoretical model of a particular heterogeneous step length distribution, the large degree of freedom can make this task extremely difficult, if not impossible. On the other hand, the fact that so many processes yield similar cover times may justify viewing all of them as essentially equivalent. This interpretation may also be relevant to the current debate regarding whether animals' movement is better represented by L\'evy Flights or by CCRW distributions  with 2 or 3 scales \cite{Pyke,Jager,Jansen,WeierstrassianMussels}. 
	Moreover, the fact that many heterogeneous step processes yield similar performances to L\'evy Flights may imply that limiting the empirical fit to either L\'evy Flights or CCRW searches  with 2 or 3 scales  may be too restrictive. Our work may suggest that instead, the focus could shift to identifying the number of scales involved in the search. 
	
	When combined with appropriate empirical measurements, 
	our lower bound  can potentially be used to indirectly show that a given intermittent process uses strictly more  than a certain number of step lengths. For example, if the process is empirically shown as a heterogeneous random walk whose cover time is almost linear, then Theorem  \ref{thm:main-lb} implies that it must use roughly logarithmic number of step lengths. From a methodological perspective, such a result  would be of particular appeal as demonstrating lower bounds in biology through mathematical arguments is extremely rare \cite{PLOS,Review}.

Finally, we note that most of the theoretical research on heterogeneous search processes which is based on differential equation techniques and computer simulations. In contrast, and similarly to \cite{LATIN},  our methodology relies on algorithmic analysis techniques and discrete probability arguments, which are more commonly used  in theoretical computer science. We believe that the computational approach presented here can contribute to a more fundamental understanding of these search processes. 
\bibliography{biblio-inter}

\newpage
\appendix
\centerline{\Huge{Appendix}}
\section{Proof of Lemma \ref{lem:main-lb}}\label{app:lem1}
	Our goal in this section is to prove Lemma \ref{lem:main-lb}.

	In what follows we stress that we count the \emph{time} and not the number of moves.
	Fix an index $0\leq i <k$. We divide time into consecutive \emph{$i$-phases}, each of time-duration precisely $L_{i+1}$ (the last one may be shorter). We next prove the following.

	\begin{claim}\label{claim:seen-per-phase} The expected number of nodes visited during the $\ell$'th $i$-phase is
		\begin{itemize}
			\item For $i=0$, $\E(N_\ell) =O\left(\sqrt{\frac{L_1}{L_0}}\right)$.
			\item For $0<i < k$,
			$\E(N_\ell) = O\left(k \sqrt{L_{i} \cdot L_{i + 1}}\right)$.
		\end{itemize}
	\end{claim}

    \begin{claimproof}{Proof of Claim \ref{claim:seen-per-phase}}
	Fix an index $i$ and consider the $i$-phases. As the last $i$-phase may be shorter and intermediate $i$-phases may start when the process is executing a jump, the value of $\E(N_\ell)$ is at most $\E(N_1)$, namely, the expected number of nodes that are visited during the first $i$-phase. 
	Let us therefore concentrate on upper bounding $\E(N_1)$. The first $i$-phases lasts during the time period $[0,L_{i+1})$. Since only endpoints of jumps are visited, if during the $i$-phase the process starts any jump of length at least $L_{i+1}$, then the number of nodes does not increase. 
	Thus, to get an upper bound on $\E(N_1)$, we may consider only trajectories that do not use such large jumps, i.e., we may restrict the process to jumps of length $L_j$, for $j\leq i$.
	
	Denote by $D$  the maximal distance achieved by the process in the time interval $[0,L_{i+1})$. We have $N_1 \leq 2D+1$. In this phase of duration $L_{i+1}$, there are at most $\frac{L_{i+1}}{L_j}$ steps of length $L_j$ that can be made, for $j\leq i$, because a jump of length $L_j$ takes $L_j$ time. Let $D_j$ be the maximal distance travelled by the jumps of length $L_j$ (when ignoring jumps of length different than $L_j$). We have $D\leq \sum_{j\leq i} D_j$. Furthermore, by the Kolmogorov's inequality, we know that the maximal distance achieved by taking $m$  random walk steps with unitary length each is of order $\sqrt{m}$, in expectation. Therefore, when the walk does steps of length $L$ the expectation of the  maximal distance after $m$ steps is $O(\sqrt{m}L)$. Thus, \[\E(D_j) = O\left( \sqrt{\frac{L_{i+1}}{L_j}}\cdot L_j \right) \] and \[ \E(N_1)\leq 2\E(D)+1 = O\left( \sum_{j\leq i}\sqrt{L_{i+1} L_j} \right) =O\left( k\sqrt{L_iL_{i+1}}\right). \]
	This establishes the second item in the claim. 
	
	Bounding the number of visited nodes $N_1$ by the distance $D$, as was done above, is not very precise, since there may be non-visited points between jumps. In order to establish the first item in the claim, i.e., the case where $i=0$, let us be more precise. In this case, we may replace the equation $N_1\leq 2D+1$ by the more precise inequality \[N_1\leq 2 \frac{D}{L_0} + 1.\] Indeed, since there are only jumps of length $L_0$, and there is no time to do a full turn of the cycle in the duration $L_1$, we visit only multiples of $L_0$. Thus, when $i=0$, we have:
	\[\E(N_1) = O\left ({\frac{\sqrt{L_1 L_0}}{L_0}} \right )= O\left (\sqrt{\frac{L_1}{L_0}} \right),\]
	as desired. This completes the proof of Claim \ref{claim:seen-per-phase}. \end{claimproof}
	
	Let us end the proof of Lemma \ref{lem:main-lb}. By Claim \ref{claim:seen-per-phase}, the number of nodes visited during the $s$ first 
	$i$-phases is
	\[
	\E\left(\sum_{\ell=1}^s N_\ell\right) \leq s \cdot O\left( E_i \right).
	\]
	where  $E_0=\sqrt{\frac{L_1}{L_0}}$ and $E_i=\sqrt{L_i L_{i+1}}$ for $1\leq i\leq k-1$.
	Next, let us set $s_1:= n \cdot \frac{c}{\cdot E_i }$ for a sufficiently small constant $c$, such that the previous bound becomes less than $n/2$.
	Using Markov's inequality, we get 
	\[\Pr\left(\sum_{\ell=1}^{s_1} N_\ell \geq n \right)< \frac{1}{2}.\]
	Therefore, with probability at least $1/2$, the process needs at least $s_1$ phases before visiting all nodes. Since the duration of a phase is $L_{i+1}$, the cover time is at least \[s_1\cdot L_{i+1} = \Omega\left(n \cdot \frac{L_{i+1}}{E_i} \right), \] 
	which is  $\Omega(n\cdot \sqrt{L_1 L_0})$ if $i=0$ and $\Omega(n\cdot\sqrt{ \frac{L_{i+1}}{L_i}})$ otherwise. This completes the proof of Lemma~\ref{lem:main-lb}. \qed

	\section{Proof of the upper bound} \label{sec:upper-bound-proof}

In this section, we prove the following theorem, for which we presented the intuition of the proof in the main text, in the case $n=b^{k}$ and, briefly, $n=\hat{n}b^{k-1}$.
 
 \ThmUp*

\subsection{Notations}
Let $V _s$ and $\xi_s$ be, respectively, the length and the sign of the $s$-th jump. More precisely, $V_s$ is a random variable taking value $L_i=b^i$ with probability $p_i=c_b b^{-i}$ for every $i \leq k-1$,  $\xi_s$ takes value $0$, $1$ or $-1$, with probabilities $\frac 12, \frac 14, \frac 14$, and the variables $(V_s)_{s\in \N}$ and $(\xi_s)_{s\in\N}$ are independent. We define the Weierstrassian random walk $Z(m)$ on $\Z$ and $X(m)$ on the cycle $C_n$, after $m$ moves, as
\begin{align}
Z(m) = \sum_{s=1}^m \xi_s\cdot V_s, \hspace{1.5cm} X(m) = Z(m) \mod n.
\label{eq:def-moves}
\end{align}
As we consider it takes one unit of time to travel a distance $1$, the time it takes to accomplish the first $m$ moves, denoted $T(m)$, is defined as
\begin{align}
\label{eq:def-time2}
T(m) := \sum_{s=1}^m  \lvert \xi_s \rvert \cdot V_s.
\end{align}
On the finite graph $C_n$, we denote  by $m_{cov}$ the (random) number of moves needed before $X$
has visited every node of $C_n$. The quantity we want to bound is $\E(T(m_{cov}))$, the expected time needed to visit all nodes, which is called the \emph{cover time}.

We also denote by $m_{hit}(x)$ the random number of moves before hitting a point $x$ for the first time. If necessary, we precise $m^{C_n}_{hit}(x)$ or $m^\Z_{hit}(x)$ to indicate the underlying topology.

Finally, the subscript $x$ in $\Pr_x$ or $\E_x$ indicates that we consider the process starting at $x$. When this subscript is absent, it means that the process starts at $0$.

\subsection{Bounding the Cover Time using Pointwise Probabilities}
In this section, we prove a few remarks that, together, establish the following lemma, that appears as Lemma \ref{lem:cov-from-pointwise} in the main text. Both Lemmas differ slightly (here we study the pointwise probabilities of $Z$ instead of $X$) as the main text presents the intuition in a simplified context.

\begin{lemma}\label{lem:cov-from-pointwise-appendix}
	If $p>0$ and $m_0>0$ are such that, for any $x\in \{0,\dots, n-1\}$,
	\begin{align}\label{eq:trick-refined-appendix}
	\frac{\sum_{m=m_0}^{2m_0}  \Pr(Z(m)=x)}{\sum_{m=0}^{m_0} \Pr(Z(m)=0)} \geq p,
	\end{align}
	then the cover time of the Weierstrassian random walk $X$ on the cycle $C_n$ is $O\left (m_0 p^{-1} k \log n \right )$.	
\end{lemma}

\begin{proof}
First, we note that $\E(V_1)$, namely, the average time taken by each non-lazy step, is roughly $k$. Specifically: 
\begin{equation}\label{eq:V1}
\E(V_1)=\sum_{i=0}^{k-1} b^i p_i =\sum_{i=0}^{k-1} c_b ({b}/{b})^i=c_b k=\Theta(k),\end{equation}
since $c_b = \frac{1-b^k}{1-b} = \Theta(1)$, as $b\geq 2$.

We next give a claim that reminds Wald's identity, but we formally prove it using the Martingale Stopping Theorem. Note that the factor $ \frac 12$ in the middle expression of the claim comes from the laziness, and hence $\E( \lvert \xi_1 V_1 \rvert ) = \frac{\E(V_1)}{2}$. 
\begin{claim}\label{lem:timeVSmoves}
	$\E(T(m_{cov})) = \E(m_{cov})\cdot \frac{\E(V_1)}{2}=\Theta(k \E(m_{cov})).$
\end{claim}
\begin{claimproof}{Proof}
Define \[Z_m:=\sum_{s\leq m} (V_s-\E(V_1) ).\] The claim is proven by showing first that $(Z_m)_m$ is a martingale with respect to $(X_m)_m$. Then, as the cover time is a stopping time for $(X_m)_m$ (i.e., the event $\{m_{cov}=m\}$ does not depend on $X_{s}$, for $s>m$), we can apply the Martingale Stopping Theorem which gives $\sum_{s\leq m_{cov}} (V_s-\E(V_1))= 0$.

In more details, recall (e.g., \cite{Mitzenmacher}[Definition 12.1]) that a sequence of random variables $(Z_m)_m$ is a {\em martingale} with respect to the sequence $(X_m)_m$ if, for all $m\geq 0$, the following conditions hold:
\begin{itemize}
    \item $Z_m$ is a function of $X_0, X_1, \dots, X_m$;
\item $\E(\lvert Z_m\rvert )<\infty$;
\item $\E(Z_{m+1} \mid X_0,\dots, X_m)  = Z_m$. 
\end{itemize}

We first claim that $Z_m=\sum_{s\leq m} (V_s-\E(V_1) )$ is a martingale with respect to $X_0,X_1,\ldots$. Indeed, since $V_s=\lvert X_s-X_{s-1}\rvert $, the first condition holds. Since $\E(\lvert Z_t\rvert )\leq \sum_{s\leq t} \E(\lvert V_s-\E(V_1)\rvert ) \leq 2t\E(V_1)<\infty$, the second condition holds. Finally, since $Z_{m+1}=Z_{m}+V_{m+1} -\E(V_1)$,  we have $\E(Z_{m+1} \mid X_0,\dots, X_m)=Z_m+\E(V_{m+1})-\E(V_1)=Z_m$, and hence the third condition holds as well.

Next, recall the Martingale Stopping Theorem (e.g., \cite{Mitzenmacher}[Theorem 12.2]) which implies that $\E(Z(T))=\E(Z_0)$, 
whenever the following three conditions hold:
\begin{itemize}
    \item $Z_0, Z_1, \dots$ is a martingale with respect to $X_0, X_1, \dots$,
    \item $T$ is a stopping time for $X_0, X_1, \dots$ such that $E(T)<\infty$, and 
    \item there is a constant $c$ such that $E(\lvert Z_{t+1} - Z_t \rvert  \mid X_0,\dots, X_t ) < c$.
    \end{itemize}

Let us prove that the conditions of the Martingale Stopping theorem hold. We have already seen that the first condition holds. Second, we need to prove that $E(m_{cov})<\infty$. This is, in fact, a general claim for an irreducible Markov chain on a discrete space (see \cite{Aldous}[Theorem 6.1] for a precise bound). Finally, we need to prove that $E(\lvert Z_{t+1} - Z_t \rvert \mid X_0,\dots, X_t )< c$ for some $c$ independent of $t$. Since $Z_{t+1} - Z_t=V_{t+1}-\E(V_1)$, we have $\E(\lvert Z_{t+1} - Z_t \rvert \mid X_0,\dots, X_t ) = \E(\lvert V_{t+1}-\E(V_1)\rvert ) \leq 2\E(V_1)$. Hence the conditions hold and the theorem gives:
\[\E(Z(m_{cov}))=\E(Z_0)=0.\]
Hence, \[ 0=\E(Z(m_{cov})) =  \E\left( - m_{cov}\E(V_1)+\sum_{s\leq m_{cov}} V_s \right)  =  - \E(m_{cov}) \E(V_1) )+\E\left(\sum_{s\leq m_{cov}} V_s\right), \]
which establishes the claim.
\end{claimproof}

\begin{claim}
	\label{lem:hit-time-from-point-probability}
Fix a positive integer $m$ and $p$. If, for any $x$ in $C_n$, $\Pr(m_{hit}(x)\leq m) \geq p$, then $\E(m_{cov})\leq mp^{-1} \log n$. 
\end{claim}
This is standard in the theory of Markov chains and stems from basic properties, but, for the sake of completeness, let us recall the proof.
\begin{claimproof}{Proof}
	By hypothesis, we have that $\Pr(m_{hit}(x) >m)\leq 1-p$ for any $x$ in the cycle. Thus: \begin{equation}\label{eq:probability-m-hit-2}
	    \Pr(m_{hit}(x) >2 m)=\Pr(m_{hit}(x) >2 m\mid m_{hit}(x) >m)\cdot \Pr(m_{hit}(x) >m).
	\end{equation} Furthermore, \begin{align*}
	    \Pr(m_{hit}(x) >2 m\mid m_{hit}(x) >m) &= \sum_{x_0\in C_n} \Pr( m_{hit}(x) > 2m \mid X_m=x_0) \Pr(X(m)=x_0) \\
	    &\geq \min_{x_0\in C_n} \Pr( m_{hit}(x) > 2m \mid X_m=x_0) \\ &= \min_{x_0\in C_n} \Pr{_{x_0}}( m_{hit}(x) > m),
	\end{align*} where we used the Markov property in the last equality.  Since we are on the cycle, and the process is symmetric, starting from $x_0$ to get to $x$ is the same as starting from $0$ to get to $x-x_0$ and therefore we have $\Pr_{x_0}( m_{hit}(x) > m)= \Pr( m_{hit}(x-x_0) > m)$. Since, by hypothesis, for any $x\in C_n$, $\Pr( m_{hit}(x) > m)$ is at most $1-p$, we have, by Eq.~\eqref{eq:probability-m-hit-2}, \[\Pr(m_{hit}(x) >2m) \leq (1-p)^2.\] 
	By a similar inductive reasoning, we obtain $\Pr(m_{hit}(x) >sm) \leq (1-p)^s$ for any node $x$ and any integer $s\geq 0$. Thus, using the definition of the expectation, we have
	
	\begin{align*}
	\E(m_{hit}(x)) &= \sum_{s=0}^\infty \Pr(m_{hit}(x)>s) =\sum_{s=0}^\infty \sum_{j=0}^{m-1} \Pr(m_{hit}(x)>sm+j) \\
	&\leq \sum_{s=0}^\infty m \Pr(m_{hit}(x)>sm)
	\leq \sum_{s=0}^\infty m (1-p)^s
	\leq {m}/{p}.
	\end{align*}
Finally, remember that Matthew's upper bound (see, e.g., \cite{Peres}[Thm 11.2]) states that: \[\E\left(m_{cov}\right)\leq \log n\cdot \max_x \E\left( m_{hit}(x)\right),\] and thus, $\E(m_{cov})$ is at most $mp^{-1}\log n$, which concludes the proof of Claim \ref{lem:hit-time-from-point-probability}.\end{claimproof}

To conclude the proof of Lemma \ref{lem:cov-from-pointwise-appendix}, it is enough to prove that, for any $x\in [0,n-1]$,
\begin{equation}\label{eq:proba-X-hits-x}
\Pr(m_{hit}(x)\leq 2m_0) \geq \frac{\sum_{m=m_0}^{2m_0}  \Pr(Z(m)=x)}{\sum_{m=0}^{m_0} \Pr(Z(m)=0)} \geq p ,\end{equation}
where the latter inequality is in fact the condition in the lemma. Lemma \ref{lem:cov-from-pointwise-appendix} is then obtained by using Claim \ref{lem:hit-time-from-point-probability}.

To establish the first inequality in Eq.~\eqref{eq:proba-X-hits-x}, note first that for any $x\in [0,n-1]$, we have $m_{hit}^{C_n}(x) \leq m_{hit}^\Z(x)$.
Indeed, if $x$ is visited by $Z$, then it is visited by $X=Z\mod n$. In particular, we have $\Pr(m_{hit}^{C_n}(x)\leq m_0) \geq \Pr( m_{hit}^\Z(x) \leq m_0)$. Hence to prove Eq.~\eqref{eq:proba-X-hits-x}, it is enough to prove \begin{equation}\label{eq:proba-hits-x}
{\Pr}(m^\Z_{hit}(x) \leq 2m_0) \geq \frac{\sum_{m=m_0}^{2m_0}  {\Pr}(Z(m)=x)}{\sum_{m=0}^{m_0} {\Pr}(Z(m)=0)}.
\end{equation} 
 For this, we rely on the following identity (see also \cite{KanadeMS16}). If $N$ is a non-negative random variable then:
\begin{align}
\Pr(N\geq 1 ) = \frac{\E ( N)}{\E(N\mid N \geq 1)}.
\label{eq:trick}
\end{align} 
We employ this identity for the random variable $N_x(m_0,2m_0)$ which is the number of times $Z$ hits $x\in \Z$ between moves $m_0$ and $2m_0$ included, for $m_0$ as in the statement of Lemma \ref{lem:cov-from-pointwise-appendix}. Note that this quantity is positive if and only if $x$ is visited during this interval, so that
\begin{equation} \label{eq:probability-hit-x-N}\Pr ( m_{hit}^\Z \leq 2m_0) \geq \Pr \left (N_x(m_0,2m_0) \geq 1 \right ).\end{equation}
Note that $N_x(m_0,2m_0)=\sum_{m=m_0}^{2m_0}\mathbbm{1}_{Z(m)=x} $. Therefore, \begin{equation}\label{eq:E_x}\E ( N_x(m_0,2m_0) )=\sum_{m=m_0}^{2m_0} \Pr(Z(m)=x).\end{equation} Note also that the denominator in Eq.~\eqref{eq:trick} applied to $N_x(m_0,2m_0)$ verifies
 \[\E_0 \left ( N_x(m_0,2m_0) \mid N_x(m_0,2m_0) \geq 1\right ) \leq \E_x \left ( N_x(0,m_0) \right ),\]
as a consequence of the Markov property, because the number of returns to $x$ is maximized whenever the first hit of $x$ is at the beginning of the time interval. Finally we have $\E_x \left ( N_x(0,m_0) \right )=\E_0 \left ( N_0(0,m_0)\right )$ because, on the line, all nodes are equivalent. Now, we have $\E_0 ( N_0(0,m_0))=\sum_{m=0}^{m_0} \Pr(Z(m)=0)$, so that
 \begin{equation}\label{eq:E_0}
     \E_0 \left ( N_x(m_0,2m_0) \mid N_x(m_0,2m_0) \geq 1\right ) \leq \sum_{m=0}^{m_0} \Pr(Z(m)=0) 
 \end{equation}
Therefore, when applied to $N_x(m_0,2m_0)$, Eq.~\eqref{eq:trick}, combined with Eqs. \eqref{eq:probability-hit-x-N} \eqref{eq:E_x} and \eqref{eq:E_0}, implies that

\begin{equation*}
{\Pr}_0(m^\Z_{hit}(x) \leq 2m_0) \geq \frac{\sum_{m=m_0}^{2m_0}  {\Pr}_0(Z(m)=x)}{\sum_{m=0}^{m_0} {\Pr}_0(Z(m)=0)},
\end{equation*} 
as desired. This establishes Eq.~\eqref{eq:proba-hits-x}, and thus completes the proof of Lemma~\ref{lem:cov-from-pointwise-appendix}. \end{proof}

\subsection{From $k$-scales search on $\Z$ to $k$ (dependent) random walks on $C_b\times \dots \times C_b \times \Z$}
This section is the conceptual core of the proof. We show how the Weierstrassian walk with $k$ scales $Z$ can be studied as $k$ dependent random walks on the space $C_b\times \dots \times C_b \times \Z$. For this we first define in Section \ref{subsubsec:defnot} the $k$ random walks $Z_0,\dots,Z_{k-1}$ on $C_b\times \dots \times C_b \times \Z$. Then, in Section \ref{subsubsec:decomposition}, we establish how the pointwise probabilities of $Z$ can be obtained by the pointwise probabilities of the $Z_i$.

\subsubsection{Definitions and Notations}\label{subsubsec:defnot}

\paragraph{Definitions.}
We define, for any $i \in [k-1]$:
\[S_i(m):= \sum_{s=1}^m \xi_s\cdot\mathbbm{1}_{( V_s = b^i)},\]
the simple (unitary) random walk on the line corresponding to the steps of length $b^i$, and
\[ J_i(m) :=  b^i S_i (m), \]
the sum of the steps of length $b^i$. Note that \begin{equation}\label{eq:Z-sum-Si}
Z(m)=\sum_{i\leq k-1} J_i(m)=\sum_{i\leq k-1} S_i(m)b^{i}.
\end{equation}We also define
\[ J_i'(m) := \sum_{j=0}^{i-1} J_j,\]
the sum of the steps of length at most $b^{i-1}$.

\paragraph{Base $b$ decomposition.}
We define, for any $x\in \Z$, the (truncated) base $b$ decomposition of $x$ as:
\[ x= \sum_{i=0}^{k-1} x_ib^i,\]
with $x_i\in [0,b-1]$ for any $i\in[0,k-2]$ and $x_{k-1}\in \Z$. This decomposition exists for any $x\in \Z$ and is unique.

\begin{remark}\label{rk:base-b}
For any $x\in \Z$, and any $i\in [0,k-2]$, we have $x_i=\lfloor xb^{-i}\rfloor \mod b$. We have also $x_{k-1}=\lfloor xb^{-(k-1)}\rfloor$. To see why, note that for any $i\in [0,k-2]$, $xb^{-i}=\sum_{j\leq k-1} x_jb^{j-i}=\sum_{j\leq i-1}x_j b^{j-i} + x_i + \sum_{j\in [i+1,k-1]}x_jb^{j-i}$, so that $\lfloor xb^{-i} \rfloor \mod b = x_i+ \lfloor \sum_{j\leq i-1}x_j b^{j-i} \rfloor  \mod b$. Since $0\leq x_j\leq b-1$, we have $0\leq \sum_{j\leq i-1}x_j b^{j} \leq b^i -1$, hence $\lfloor xb^{-i} \rfloor \mod b = x_i$. For $i=k-1$, the proof is similar, except we do not need to take modulo $b$ (as $x_{k-1}\in \Z$).
\end{remark}

\paragraph{Decomposition of $Z$ in the base $b$.}
 In this base, let us denote by $Z_i$ the $i$-th coordinate of $Z$, so that:
\[ Z(m)=\sum_{i=0}^{k-1} Z_i(m)b^i.\]
By Remark \ref{rk:base-b} and  Eq.~\eqref{eq:Z-sum-Si}, we have, for $i\leq k-2$,
\begin{align*} Z_i(m)&=Z(m)b^{-i}\mod b= \sum_{j\leq k-1} S_i(m)b^{j-i} \mod b = \sum_{j\leq i} S_i(m)b^{j-i} \mod b \\&= R_i(m)+N_i(m)\mod b\end{align*}
where we define, for $i\leq k-2$, \[R_i(m):=S_i(m)\mod b,\] and \[N_i(m):=\left \lfloor \left (\sum_{j\leq i-1}S_j(m)b^{j}\right)b^{-i} \right \rfloor  \mod b=\lfloor J'_i(m)b^{-i}\rfloor \mod b.\] Similary, we decompose $Z_{k-1}(m)$ as the sum of $R_{k-1}(m)=S_{k-1}(m)$ and \begin{equation}\label{eq:def-N-k-1}
N_{k-1}(m)=\lfloor J'_{k-1}(m) b^{-(k-1)}\rfloor .
\end{equation}
$R_i$ corresponds to the steps of length $b^i$ and is a lazy random walk on $C_b$ that moves with probability $\frac{p_i}{2}$. $N_i$ can be thought of as the noise from smaller coordinates. For instance, if $Z_{i-1}(m)$ is $b-1$ and a step of length $b^{i-1}$ is done, then $N_i(m)$ will be incremented of one. Note that $N_0(m)=0$ always, and that the $k-1$'st coordinate is defined on $\Z$, thus $R_{k-1}$ and $N_{k-1}$ are not defined modulo $b$.

So far, we have decomposed $Z$ as a linear combination of $k-1$ simple, dependent, random walks. Now we will define additional variables that will allow to control the dependencies between the $Z_i$.


\paragraph{Number of steps of length $b^i$.}

We denote by $M_i(m)$ the number of steps of length $b^i$ done up to move $m$, i.e., \[M_i(m)=\sum_{s\leq m} \mathbbm{1}_{V_s=b^i}.\] This random variable follows a binomial distribution with parameter $p_i=c_b b^{-i}$, and is thus concentrated around its mean: \[\mu_i:=mp_i=c_b m b^{-i}=\Theta(mb^{-i}).\] Precisely, we will show that $M_i(m)$ is likely to belong to the interval: \begin{equation}\label{eq:def-Q-i} Q_i:=\begin{cases}[\frac 12 \mu_i, \frac 32 \mu_i], \mbox{~for $i>0$, and } \\ [\frac 14 m,m] \mbox{~for $i=0$.}\end{cases}\end{equation}

\subsubsection{From $Z$ to the $Z_i$} \label{subsubsec:decomposition}
\paragraph{Fixing the number of steps $M_i(m)$.}
Here, we look at what happens when we fix the number of steps of length $b^i$, $M_i(m)$, to be $q_i$. We start with the following important remark.

\begin{remark}\label{rk:conditioning}
	In general, the variables $R_i(m)$ and $N_i(m)$ are dependent. For example if $M_1(1)=1$, then $M_0(1)=0$, since we choose only one step-length between $0$ and $1$. However, once we condition on $M_i(m)=q_i$, $R_i(m)$ and $N_i(m)$ become independent and $R_i(m)$ has then the law of a lazy (with parameter $\frac 12$) random walk after $q_i$ steps. I.e., we have, for any $y\in C_b$, or $y\in \Z$ if $i=k-1$, and any $q_i\leq m$,
	\begin{equation}\label{eq:R-cond-M} \Pr\left(R_i(m)=y \mid M_i(m)=q_i\right )=p_{q_i}^{G_i}(y)\end{equation}
	where 
	\[G_i=\begin{cases} C_b \text{ if } i\in [0,k-2] \\
	\Z \text{ if } i=k-1,\end{cases} \]
	and $p_{q_i}^{G}(y)$ is the law of a lazy (with parameter $\frac 12$) random walk on $G\in \{\Z, C_b\}$, that starts at $0$, to visit the node $y$ at step $q_i$.
\end{remark}

Considering this remark, we write, with $m\geq0$ and $x=\sum_{j=0}^{k-1}x_j b^j\in\Z$,
\begin{align}\label{eq:Z-conditioning-M} \Pr(Z(m)=x)=\sum_{q_0+\dots+ q_{k-1}=m}\mathcal{P}_{x,\mathbf{q}} \cdot \mathcal{M}_{\mathbf{q}},\end{align}
where $\mathbf{q}=(q_0,\dots,q_{k-1})$, \[\mathcal{P}_{x,\mathbf{q}}=\Pr\left (Z(m)=x\mid \forall t\leq k-1, M_{t}(m)= q_t \right ), \] and \[ \mathcal{M}_{\mathbf{q}}=\Pr \left (\forall t\leq k-1, M_{t}(m)= q_t \right ) .\]
Since the base $b$ decomposition is unique, we have $Z(m)=x$ if and only if $Z_s(m)=x_s$ for all $s\leq k-1$. Hence, \begin{align*}
\mathcal{P}_{x,\mathbf{q}}= \prod_{s=0}^{k-1} \Pr\left ( Z_s(m)=x_s\mid \mathcal{A}_{s,x,\mathbf{q}} \right ),\end{align*}
where $\mathcal{A}_{s,x,\mathbf{q}}$ denotes the event  $(\forall j<s, Z_{j}(m)=x_{j}) \cap (\forall t\leq  k-1, M_{t}(m)= q_t) $.

Since $Z_i(m)=R_i(m)+N_i(m)=x_i$ if and only if $R_i(m)=x_i-y$ and $N_i(m)=y$ for some $y\in C_b$ ($\Z$ if $i=k-1$), using Remark \ref{rk:conditioning}, we have:
\begin{align}\label{eq:decomposition-ZRN}
\mathcal{P}_{x,\mathbf{q}}=\prod_{s=0}^{k-1} \sum_{y} \Pr\left ( R_s(m)=x_s-y\mid \mathcal{A}_{s,x,\mathbf{q}} \right )\cdot \Pr\left (N_s(m)=y\mid \mathcal{A}_{s,x,\mathbf{q}}. \right). \end{align}
Using Eq.~\eqref{eq:R-cond-M}, we have:
\begin{align*}\Pr\left ( R_s(m)=x_s-y\mid \mathcal{A}_{s,x,\mathbf{q}} \right )=p^{G_s}_{q_s}(x_s-y).
\end{align*}
Inserting this in Eq.~\eqref{eq:decomposition-ZRN}, we obtain that
\begin{align}\label{eq:decomposition-ZRN2}
\mathcal{P}_{x,\mathbf{q}} =\prod_{s=0}^{k-1} \sum_{y} p^{G_s}_{q_s}(x_s-y)\cdot \Pr\left (N_s(m)=y\mid \mathcal{A}_{s,x,\mathbf{q}} \right). \end{align}
Hence, in Eq.~\eqref{eq:Z-conditioning-M}, we have:
\begin{align}\label{eq:Z-conditioning-M-full} \Pr(Z(m)=x)=\sum_{q_0+\dots+ q_{k-1}=m} \mathcal{M}_{\mathbf{q}}\cdot \prod_{s=0}^{k-1} \sum_{y} p^{G_s}_{q_s}(x_s-y)\cdot \Pr\left (N_s(m)=y\mid \mathcal{A}_{s,x,\mathbf{q}} \right).\end{align}

\paragraph{Bounds on pointwise probabilities of $Z$ taking dependencies between the coordinates into account.}
Recall, with Lemma \ref{lem:cov-from-pointwise-appendix}, that we need only bounds on $\Pr(Z(m)=0)$ and $\Pr(Z(m)=x)$ for $x\in[n]$ to estimate the cover time of the Weierstrassian random walk on $C_n$. In the following two lemmas, we show how such bounds can be obtained by the independent study of:
\begin{itemize}
	\item the distributions of $M_i(m)$, studied in Section \ref{subsubsec:number-steps}
	\item the probability  $p_q^{G_i}(y)$. It is given in Section \ref{subsubsec:rw-distributions}, and
		\item the noise in the last coordinate, $N_{k-1}(m)$, studied in Section \ref{subsubsec:noise}.
\end{itemize}

Note that, when neglecting the dependencies between the coordinates, and assuming that $M_i(m)$ is exactly its expected value $mp_i$, we have, as detailed in the main text,
\[ \Pr(Z(m)=x) = \prod_{s=0}^{k-1} \Pr(R_s(m)=x_s) = \prod_{s=0}^{k-1} p_{mp_s}^{G_s}(x_s). \]
Note also that we have $\prod_{s=0}^{k-1} p_{mp_s}^{G_s}(x_s)\leq \prod_{s=0}^{i} p_{mp_s}^{G_s}(x_s)$ for any $i\leq k-1$. This is useful in particular when $mp_{i+1} \leq 1 \leq mp_i$, i.e. when $m\in [c_b b^i, c_b b^{i+1}]$. The following two lemmas provide the additional components that appear when taking into account the noise and the fact that the number of steps $M_i(m)$ does not always equal its expected mean $mp_i$. We shall first prove the following upper bound.

\begin{lemma}\label{lem:separation-ub} For any $m\geq 0$ and any $i\leq k-1$, 
	\begin{equation}\label{eq:separation-ub} 	
		 \Pr(Z(m)=0) \leq \prod_{s=0}^{i} \max_{y,q_s \in Q_s} p_q^{G_s}(y) + \sum_{j=0}^i \left ( \Pr(M_j(m)\notin Q_j)  \prod_{s=0}^{j-1} \max_{y,q_s \in Q_s} p_q^{G_s}(y)\right ) \end{equation}
\end{lemma}
We will prove in Section \ref{subsubsec:returns-to-0}, that the dominating term of this upper bound is \[ \prod_{s=0}^{i} \max_{y,q_s \in Q_s} p_q^{G_s}(y), \] as is hinted by the intuition. We will also prove the following lower bound. It uses the event $\mathcal{A}_{k-1,x,\mathbf{q}}$ that we recall as $(\forall j\leq k-2, Z_{j}(m)=x_{j}) \cap (\forall t\leq  k-1, M_{t}(m)= q_t) $.

\begin{lemma}\label{lem:separation-lb}
For any $m\geq 0$, any $x\in \Z$, and any $I$ interval of $\Z$,
	\begin{align}\label{eq:separation-lb} \Pr ( Z(m)=x ) \geq & \min_{\forall i, q_i\in Q_i}\Pr(N_{k-1}(m)\in I\mid \mathcal{A}_{k-1,x,\mathbf{q}}) \\ &\cdot \min_{y\in I, q\in Q_{k-1}} p^{\Z}_q(x_{k-1}-y) \cdot\prod_{s=0}^{k-2}  \min_{y\in C_b, q\in Q_s} p_q^{C_b}(y)\nonumber \\ \nonumber &\cdot \Pr(\forall j\leq k-1, M_j(m)\in Q_j ).\end{align}
\end{lemma}
We will prove in Section \ref{subsubsec:returns-to-x}, that the dominating term of this lower bound, when $m\geq b^k$, is the second one, namely, $ \min_{y\in I, q\in Q_{k-1}} p^{\Z}_q(x_{k-1}-y) \cdot\prod_{s=0}^{k-2}  \min_{y\in C_b, q\in Q_s} p_q^{G_s}(y) $. Indeed, with $I$ well-chosen, and for $m\geq b^k$, we will prove that the first and last factors are $\Omega(1)$, in Sections \ref{subsubsec:noise} and \ref{subsubsec:number-steps}, respectively.
		\begin{proof}[Proof of Lemma \ref{lem:separation-ub}]
	We start with Eq.~\eqref{eq:Z-conditioning-M-full}
	\begin{align}\nonumber  \Pr(Z(m)=x)&=\sum_{q_0+\dots+ q_{k-1}=m} \mathcal{M}_{\mathbf{q}}\cdot \prod_{s=0}^{k-1} \sum_{y} p^{G_s}_{q_s}(x_s-y)\cdot \Pr\left (N_s(m)=y\mid \mathcal{A}_{s,x,\mathbf{q}} \right) \\ \nonumber &\leq \sum_{q_0+\dots+ q_{k-1}=m} \mathcal{M}_{\mathbf{q}}\cdot \prod_{s=0}^{k-1} \max_{y} p^{G_s}_{q_s}(y) \sum_{y} \Pr\left (N_s(m)=y\mid \mathcal{A}_{s,x,\mathbf{q}} \right) \\ &\leq \sum_{q_0+\dots+ q_{k-1}=m} \mathcal{M}_{\mathbf{q}}\cdot \prod_{s=0}^{k-1} \max_{y} p^{G_s}_{q_s}(y), \label{eq:Z-dec-up} \end{align}
	where we used in the last inequality that $\sum_{y} \cdot \Pr\left (N_s(m)=y\mid \mathcal{A}_{s,x,\mathbf{q}} \right) = 1$. As the number of steps of length $b^j$, $M_j(m)$, is likely to belong to $Q_j$ (defined by Eq.~\eqref{eq:def-Q-i}), we make the following decomposition of the sum in Eq.~\eqref{eq:Z-dec-up}, for any $i\leq k-1$:
	\begin{align}\label{eq:dec-sum}  \sum_{q_0+\dots+ q_{k-1}=m}=\sum_{\substack{q_0+\dots+q_{k-1}=m \\ q_0 \in Q_0, \ldots, q_i \in Q_i}} +\sum_{j=0}^i \sum_{\substack{q_0+\dots+q_{k-1}=m \\ q_0 \in Q_0, \ldots, q_{j-1} \in Q_{j-1}, q_j \notin Q_j}}  \end{align}
	 The  intuition behind this decomposition is that when $q_0,\dots,q_i\in Q_0\times\dots\times Q_i$, we may obtain a good bound on the pointwise probability of the coordinates $0$ to $i$, giving an upper bound on $\prod_{s=0}^{k-1} \max_{y} p^{G_s}_{q_s}(y)$ (bounding the factors for $s>i$ by $1$). When for some $j\leq i$, $q_0 \in Q_0, \ldots, q_{j-1} \in Q_{j-1},q_j \notin Q_j$, we have such a bound for the coordinates $0$ to $j-1$, yielding a (weaker) bound on $\prod_{s=0}^{k-1} \max_{y} p^{G_s}_{q_s}(y)$. To compensate for this weaker bound, we use that the event $M_j(m)\notin Q_j$ is unlikely, to get a bound on $\mathcal{M}_{\mathbf{q}}.$ 
	
	Let us first consider the inner sum in the second sum of Eq.~\eqref{eq:dec-sum}. We have:
	\begin{align*}
	\sum_{\substack{q_0+\dots+q_{k-1}=m \\ q_0 \in Q_0, \ldots, q_{j-1} \in Q_{j-1}, q_j \notin Q_j}} \mathcal{M}_{\mathbf{q}}\cdot \prod_{s=0}^{k-1} \max_{y}p_{q_s}^{G_s}(y)
	\\ \leq  \left( \prod_{s=0}^{j-1}\max_{y,q_j\in Q_s} p_{q_s}^{G_s}(y)\right)\cdot \sum_{\substack{q_0+\dots+q_{k-1}=m \\ q_0 \in Q_0, \ldots, q_{j-1} \in Q_{j-1}, q_j \notin Q_j}} \mathcal{M}_{\mathbf{q}} \\ 
	\leq \left( \prod_{s=0}^{j-1}\max_{y,q_s\in Q_s} p_{q_s}^{G_s}(y) \right) \cdot \Pr\left (M_0(m)\in Q_0,\dots,M_{j-1}(m)\in Q_{j-1},M_j(m)\notin Q_j \right ) \\ 
	\leq \left ( \prod_{s=0}^{j-1}\max_{y,q_s\in Q_s} p_{q_s}^{G_s}(y)\right) \cdot \Pr\left (M_j(m)\notin Q_j \right ). 
	\end{align*}
	By similar computations, we bound the first sum: \[\sum_{\substack{q_0+\dots+q_{k-1}=m \\ q_0 \in Q_0, \ldots, q_i \in Q_i}}\mathcal{P}_{x,\mathbf{q}} \cdot \mathcal{M}_{\mathbf{q}}\leq \prod_{s=0}^{i} \max_{y,q_s \in Q_s} p_{q_s}^{G_s}(y).\] Inserting into Eq.~\eqref{eq:Z-dec-up}, we get:
	\begin{align*} \Pr(Z(m)=x) \leq \prod_{s=0}^{i} \max_{y,q_s \in Q_s} p_q^{G_s}(y)+\sum_{j=0}^i \Pr(M_j(m)\notin Q_j)  \prod_{s=0}^{j-1} \max_{y,q_s \in Q_s} p_{q_s}^{G_s}(y) ,   \end{align*}
	as desired.
\end{proof}

	\begin{proof}[Proof of Lemma \ref{lem:separation-lb}]	
Let us recall Eq.~\eqref{eq:Z-conditioning-M-full}:
	\begin{align}\nonumber  \Pr(Z(m)=x)=\sum_{q_0+\dots+ q_{k-1}=m} \mathcal{M}_{\mathbf{q}}\cdot \prod_{s=0}^{k-1} \sum_{y} p^{G_s}_{q_s}(x_s-y)\cdot \Pr\left (N_s(m)=y\mid \mathcal{A}_{s,x,\mathbf{q}} \right) \\ \nonumber \geq \sum_{\substack{q_0+\dots+q_{k-1}=m \\ q_0 \in Q_0, \ldots, q_{k-1} \in Q_{k-1}}}\mathcal{M}_{\mathbf{q}}\cdot \prod_{s=0}^{k-1} \sum_{y\in I_s} p^{G_s}_{q_s}(x_s-y)\cdot \Pr\left (N_s(m)=y\mid \mathcal{A}_{s,x,\mathbf{q}} \right),\end{align}
	where $I_s=C_b$ for $s\leq k-2$ and $I_{k-1}=I$ is any interval of $\Z$. We then lower bound $p^{G_s}_{q_s}(x_s-y)$ by $\min_{y\in I_s} p^{G_s}_{q_s}(x_s-y)$, and use that $\sum_{y\in I_s} \Pr\left (N_s(m)=y\mid \mathcal{A}_{s,x,\mathbf{q}} \right) =\Pr(N_s(m)\in I_s \mid \mathcal{A}_{s,x,\mathbf{q}})$, which is $1$ for $s\leq k-2$, and $\Pr(N_{k-1}(m)\in I\mid \mathcal{A}_{k-1,x,\mathbf{q}})$ for $s=k-1$, to get:
	\begin{align*} \Pr(Z(m)=x)\geq  \sum_{\substack{q_0+\dots+q_{k-1}=m \\ q_0 \in Q_0, \ldots, q_{k-1} \in Q_{k-1}}}\mathcal{M}_{\mathbf{q}} \cdot \Pr\left (N_{k-1}(m)\in I \mid \mathcal{A}_{k-1,x,\mathbf{q}} \right) \cdot \prod_{s=0}^{k-1}  \min_{y\in I_s} p^{G_s}_{q_s}(x_s-y) \\ 
	\geq \min_{\forall i, q_i\in Q_i} \left \{ \Pr\left (N_{k-1}(m)\in I \mid \mathcal{A}_{k-1,x,\mathbf{q}} \right) \cdot \prod_{s=0}^{k-1}  \min_{y\in I_s} p^{G_s}_{q_s}(x_s-y) \right \} \cdot
	\sum_{\substack{q_0+\dots+q_{k-1}=m \\ q_0 \in Q_0, \ldots, q_{k-1} \in Q_{k-1}}}\mathcal{M}_{\mathbf{q}}  \end{align*}
	To conclude, we use the definition of $\mathcal{M}_{\mathbf{q}}$ to see that \[\sum_{\substack{q_0+\dots+q_{k-1}=m \\ q_0 \in Q_0, \ldots, q_{k-1} \in Q_{k-1}}}\mathcal{M}_{\mathbf{q}} = \Pr\left (M_0(m)\in Q_0,\dots,M_{k-1}(m)\in Q_{k-1} \right ).\]
\end{proof}

\subsection{Estimating the terms in Lemmas \ref{lem:separation-ub} and \ref{lem:separation-lb}}
In order to estimate the terms in Lemmas \ref{lem:separation-ub} and \ref{lem:separation-lb}, we need to understand \begin{itemize}
	\item the distribution of $M_i(m)$,
	\item the distribution of $p_{q}^{G_s}$.
	\item the distribution of the noise $N_{k-1}(m)$,
\end{itemize}
They will be studied in Sections  \ref{subsubsec:number-steps}, \ref{subsubsec:rw-distributions}, and  \ref{subsubsec:noise}, respectively.
But first, let us start with a very technical claim.

\subsubsection{Preliminary technical computations}
In what follows, we will use several times the following technical claim.
\begin{claim}\label{lem:prod-exponentials}
For any $i\geq 0$, and any constants $c\in(0,1)$ and $c'>0$, we have \[ \prod_{s=0}^i \left( 1 - ce^{-c'b^{i-s}}\right) = \Theta(1), \quad \text{ and } \quad  \prod_{s=0}^i \left( 1 + ce^{-c'b^{i-s}}\right) = \Theta(1).\]\end{claim}
\begin{claimproof}{Proof}
   Let us consider the first product. Remark that it is upper bounded by $1$. For the lower bound, as $c<1$ all terms are positive and we can take its logarithm,
  \[ \sum_{s=0}^i \log \left( 1 - ce^{-c'b^{i-s}}\right), \]
  which is negative as $c>0$. To lower bound it, we upper bound its absolute value. For this, we use that $e^{-c'b^{i-s}} \leq e^{-c'}<1$ and $-\log ( 1-t) = O(t )$ for $t\in(0,e^{-c'})$ to get:
    \[ -\sum_{s=0}^i \log \left( 1 - ce^{-c'b^{i-s}}\right) = O\left(  \sum_{s=0}^i e^{-c'b^{i-s}}\right).  \]
    Then, use that $e^{-c't}=O(t^{-1})$ for any $t>0$ to get:     \[ -\sum_{s=0}^i \log \left( 1 - ce^{-c'b^{i-s}}\right) = O\left(  \sum_{s=0}^i b^{s-i} \right) = O\left(  \sum_{s=0}^i b^{-s}\right)=O(1).  \]
    Taking the opposite of this, and then the exponential, proves the first part of Claim~\ref{lem:prod-exponentials}. The second part is done similarly.
\end{claimproof}

\subsubsection{Concentration of $M_i(m)$ around its mean $mp_i$} \label{subsubsec:number-steps}
For any $i\leq k-1$, and any $m\geq 1$, $M_i(m)$ follows a binomial distribution with parameter $p_i$ and is thus concentrated around its mean $mp_i$. Since $Q_i=[\frac 12 mp_i, \frac 32 mp_i]$ for $i>0$ and $Q_0=[\frac 14 m,m]$, we can use Chernoff's bound (Theorems 4.4 and 4.5 in \cite{Mitzenmacher}) to obtain:
\begin{equation}\label{eq:M_i-concentration}\Pr \left ( M_i (m)\notin Q_i \right )  \leq e^{-cmp_i}=e^{-cc_b m b^{-i}},
\end{equation}
for some constant $c>0$. This is the basis for the following lemma, which will essentially ensure that, for $m\geq b^k$, we can suppose that, for all $i\leq k-1$, $M_i(m)\in Q_i$

\begin{lemma}\label{lem:multinomial-concentration}
There are positive constants $c'$ and $c''$ such that for $m\geq c' b^k$, \[\Pr(\forall i\leq k-1, M_i(m)\in Q_i)>c''.\]
\end{lemma} 

\begin{proof}[Proof of Lemma \ref{lem:multinomial-concentration}]
	
Using the union bound and Eq.~\eqref{eq:M_i-concentration}, we get:
	\begin{align*}\Pr(\exists i\leq k-1, M_i(m)\notin Q_i) &\leq \sum_{i\leq k-1}\Pr(M_i(m)\notin Q_i) \\
	&\leq \sum_{i\leq k-1} e^{-cc_bmb^{-i}} 
	\leq \sum_{i\leq k-1} e^{-cc_b c'b^{k-i}} \\
	&\leq \frac{1}{ecc'c_b}\sum_{i\leq k-1} \frac{1}{b^{k-i}} \\ &\leq \frac{1}{ecc'c_b} \frac{1-b^{-k}}{b-1}\leq \frac{2}{ecc'},\\
	\end{align*}
	where we used that $m\geq c' b^k$ and $e^{-t}=O(\frac 1t)$ for $t>0$. For $c'$ well-chosen, this is less than $1-c''$ with $c''>0$. Hence, we have:
	\[\Pr(\forall i\leq k-1, M_i(m)\in Q_i)=1-\Pr(\exists i\leq k-1, M_i(m)\notin Q_i) \geq c'', \]
	as claimed by Lemma~\ref{lem:multinomial-concentration}.
\end{proof}

\subsubsection{Random walks distributions}\label{subsubsec:rw-distributions}
We need to recall estimations for the distribution of a random walk over the infinite line, and over the cycle $C_b$. Since the random walk over the cycle is obtained by projecting the random walk on $\Z$ modulo $b$, let us first state the results on $\Z$.

\begin{claim}\label{claim:lrw-line-distribution}
For a $\frac 12$-lazy random walk on $\Z$ that begins at $0$, we have, for any $q\geq 1$, and any $y\in \Z$, the probability to visit $y$ at step $q$ is:
\[ p_{q}^{\Z}(y) \leq cq^{-\frac 12},\]
with $c>0$ some constant. Furthermore, for any constant $c''>0$, there is a constant $c'>0$ such that for any $y\in [-c''\sqrt{q},c''\sqrt{q}]$, we have
\[  p_{q}^{\Z}(y) \geq c'q^{-\frac 12}. \]
 \end{claim}

\begin{claimproof}{Proof}It is easy to prove that, due to the laziness of parameter $\frac 12$, we have $p_{q}^{\Z}(y)\geq p_{q}^{\Z}(y+1)$ for any $y\geq 0$. Hence we can restrict what follows to $y=O(\sqrt{q})$. In this case, the bounds in \cite{Lawler}[Proposition 2.5.3] show that the distribution of a non-lazy random walk on $\Z$ is of order $\Theta(q^{-\frac 12})$. Going from there to a lazy random walk that moves with probability $\frac 12$, we just need to apply again a concentration argument for a Bernoulli variable. This allows to link the behaviour of the lazy random walk with $m$ steps with that of the non-lazy random walk with $\Theta(m)$ steps.
\end{claimproof}

\begin{claim}\label{claim:lrw-cycle-distribution}
For a $\frac 12$-lazy random walk on $C_b$ that begins at $0$, we have, for any $q\geq 1$, and any $y\in C_b$:
\[ p_{q}^{C_b}(y) \leq \begin{cases}
 cq^{-\frac 12}, \quad q\leq b^2 \\ b^{-1}(1+ce^{-c'qb^{-2}}), \quad q\geq b^2\end{cases}\]
where $c$ and $c'$ are positive constants. Furthermore there are constants $c''\in(0,1)$ and $c'''>0$ such that for any $q\geq b^2$,
 \[ p_{q}^{C_b}(y) \geq b^{-1}(1-c''e^{-c'''qb^{-2}}). \]
 \end{claim}
 Note that $c''<1$ ensures that this lower bound (which holds for all $q\geq b^2$) is at least $\Omega(\frac 1b)$.
\begin{claimproof}{Proof}
First, the upper bound simply follows as a particular case of the distribution of a random walk in regular graphs \cite{Aldous}[Prop 6.18]. 

The lower bound requires more explanation. Informally, it stems from the mixing properties of the cycle. Recall that the mixing time of the cycle is $\Theta(b^2)$, which means that after this time, the nodes have probability roughly $\frac 1b$ to be visited. In what follows, we make this statement more precise.

Define the {\em separation distance} as: \[s(q)=\min_{y\in C_b} \{ 1-b\cdot p_q^{C_b}(y)\}=\inf \{ s : p_q^{C_b}(y)\geq \frac{1-s}{b}, \forall y \in C_b \},\]
and the {\em total variation distance} as: \[d(q)=\frac 12 \sum_{y\in C_b} \lvert p_q^{C_b}(y)-\frac 1b\rvert.\] 
We have, as a consequence of the mixing time of the cycle being less than $b^2$, that $d(q)\leq \eps$ for $q\geq b^2 \log (\eps^{-1})$ (see \cite{Peres}[5.3.1 and Eq. (4.36)]). Furthermore, by \cite{Peres}[Lemma 19.3 and Eq.~(4.24)], we have $s(2q)\leq 1-(1-2d(q))^2$ for any $q\geq 1$. Hence, for $q\geq 2b^2 \log (\eps^{-1})$, we have $s(q)\leq 1-(1-2\eps)^2=4\eps - 4 \eps^2 < 4\eps$. That is, when $q\geq 2b^2 \log (\eps^{-1})$, we have, for any $y\in C_b$: \[p_q^{C_b}(y) \geq  \frac 1b(1-4\eps).\] With the change of variable $\eps= \exp{(-\frac{q}{2b^2})}$, we have \[p_q^{C_b}(y) \geq  \frac 1b(1-4\exp{(-\frac{q}{2b^2})}),\] 
which is not meaningful (as the bound is negative)  when  $q\leq 2b^2$. In fact, we will use this bound only for $q\geq Cb^2$, with $C=2\log(8)>1$. This ensures that $1-4\exp{(-\frac{q}{2b^2})}\geq \frac 12$ which makes for a more useful lower bound.

Now, for  $b^2\leq q \leq Cb^2$, we can lower bound $p_q^{C_b}(y)$ by $p_q^{\Z}(y)$, and use Claim \ref{claim:lrw-line-distribution}, to show that $p_q^{C_b}(y)\geq C'\frac 1b$ for some $C'\in(0,1)$. Altogether, we have $p_q^{C_b}(y) \geq  \frac 1b F(q)$ for any $q\geq b^2$, where: 
\[ F(q):=\begin{cases} C' \text{ for } q\in [b^2,Cb^2] \\
1-4\exp{(-\frac{q}{2b^2})} \text{ for } q \geq Cb^2.\end{cases} \]
To conclude, we need to show that we can bound $F(q)$ from below, for all $q\geq b^2$, by $(1-c''e^{-c'''qb^{-2}})$, for a good choice of $c''\in (0,1)$ and $c'''>0$. This is equivalent to:
establishing that: \[ \begin{cases} c''e^{-c'''qb^{-2}} \geq 1-C' \text{ for } q\in [b^2,Cb^2] \\
c''e^{-c'''qb^{-2}} \geq 4\exp{(-\frac{q}{2b^2})} \text{ for } q \geq Cb^2,\end{cases} \]
which is in turn equivalent to:
\[ \begin{cases} c''e^{-c'''qb^{-2}} \geq 1-C' \text{ for } q\in [b^2,Cb^2] \\
c''e^{qb^{-2}(\frac 12 - c''')} \geq 4 \text{ for } q \geq Cb^2.\end{cases} \]
Since we are looking for $c''<1$, for the second condition to be true, we need that $c'''<\frac 12$ (otherwise, it is obvious that the condition will not hold for $q\rightarrow \infty$). Given $c'''<\frac 12$, the left hand side of the second equation is increasing with $q\geq Cb^2$ and thus it is enough to verify the condition at $q=Cb^2$. Similarly, the left hand side of the first equation is decreasing with $q$ and thus it is enough to verify the condition at $q=Cb^2$. 
The system is thus equivalent to:
\[ \begin{cases} c''e^{-c'''C} \geq 1-C' \\
c''e^{C(\frac 12 - c''')} \geq 4\end{cases}, \]
which is in turn equivalent to the condition $c''e^{-c'''C}\geq M$ for $M:=\max\{1-C',4e^{-\frac{C}{2}}\}$. Since $C=2\log 8$, we have $M:=\max\{1-C',\frac 12\}$. Since $M<1$, we may take $c''=\frac{1+M}{2}<1$. Then it suffices to take $c'''$ small enough, e.g., $c'''=\frac{1}{C} \log(\frac{1}{2}+\frac{1}{2M}) >0$. 
With these parameters, we have proved:
\[ p_q^{C_b}(y) \geq  \frac 1b F(q) \geq \frac 1b ( 1-c''e^{-c'''qb^{-2}}),\]
for any $q\geq b^2$, and with $c''<1$. This concludes the proof of Claim \ref{claim:lrw-cycle-distribution}.
\end{claimproof}

With Claims \ref{claim:lrw-line-distribution} and \ref{claim:lrw-cycle-distribution}, we can obtain the following Lemma. Intuitively, Lemma \ref{lem:prod-proba-rw} gives the distribution of $(R_0,\dots, R_{k-1})$ when they are approximated as independent. As we will show, the bounds of Lemma \ref{lem:prod-proba-rw} are good approximations of the distributions of $\Pr(Z(m)=0)$.
\begin{lemma}\label{lem:prod-proba-rw}
We have, for any $m\geq b^k$, any $x\in \Z$,
\begin{equation}\label{eq:prod-proba-rw-lb}
\prod_{s=0}^{k-2} \min_{y\in C_b, q\in Q_{s}} p_{q}^{G_s}(x_s) = \Omega(b^{-(k-1)}).
\end{equation}
We have also, for any $i\leq k-1$, $m\in (b^i,b^{i+1}]$,
\begin{equation}\label{eq:prod-proba-rw-ub}
    \prod_{s=0}^j \max_{y,q \in Q_s} p_{q}^{G_s} (y) = \begin{cases}  O(b^{-j-1}) \text{ if } j\leq i-2, \\
    O\left ( \frac{1}{\sqrt{mb^{i-1}}} \right ) \text{ if } j=i-1, \\ 
    O(\frac{\sqrt{b}}{m}) \text{ if } j=i,\end{cases}
\end{equation}
and, for any $m\geq b^k$,
\begin{equation}\label{eq:prod-proba-rw-ub'}
    \prod_{s=0}^{k-1} \max_{y,q \in Q_s} p_{q}^{G_s} (y) = O\left ( \frac{1}{\sqrt{mb^{k-1}}} \right ). 
\end{equation}
\end{lemma}

\begin{proof}
Let us show first Eq.~\eqref{eq:prod-proba-rw-lb}. For $j\leq k-2$, $q\in Q_j$ and $m\geq b^k$, we have $q=\Theta(mp_i)=\Theta(mb^{-i})=\Omega(b^2)$. Applying the lower bound in Claim \ref{claim:lrw-cycle-distribution}, we have, for some constants $c\in (0,1)$ and $c'$,
\begin{align*} \prod_{j=0}^{k-2} \min_{y\in C_b, q\in Q_{j}} p_{q}^{G_j}(x_j) \geq \prod_{j=0}^{k-2}  \min_{q\in Q_{j}} \left ( b^{-1}(1-ce^{-c'qb^{-2}}) \right ) \\ 
\geq b^{-(k-1)} \prod_{j=0}^{k-2} \left ((1-ce^{-c'c_bmb^{i-2}}) \right) \geq  b^{-(k-1)}\prod_{j=0}^{k-2} \left ((1-ce^{-c'c_bb^{k-2-i}}).\right) \end{align*}
We conclude by applying Lemma \ref{lem:prod-exponentials} to show that $\prod_{j=0}^{k-2} (1-ce^{-c' c_b b^{k-2-i}})=\Omega(1)$.

To prove Eq.~\eqref{eq:prod-proba-rw-ub}, we proceed similarly. Let $i \leq k-1$ and $m\in (b^i,b^{i+1}]$. Using this time the upper bound from Claim \ref{claim:lrw-cycle-distribution}, we have, for $j\leq i-2$,
\begin{align} \prod_{s=0}^{j} \max_{y\in C_b, q\in Q_{s}} p_{q}^{G_s}(y) \leq \prod_{s=0}^{j}  \max_{q\in Q_{s}} \left ( b^{-1}(1+c''e^{-c'''qb^{-2}}) \right ) = O(b^{-j-1}) \label{eq:prod-prob-j} \end{align}
where the last equality is justified as above. For the cases $j=i-1$, by the upper bound in Claim \ref{claim:lrw-cycle-distribution}, for $q\in Q_{i-1}$, we have $\max_{y} p^{(i-1)}_q(y) = O ( \sqrt{\frac{b^{i-1}}{m}})$. Using Eq.~\eqref{eq:prod-prob-j}, we then have
\begin{align} \prod_{s=0}^{i-1} \max_{y\in C_b, q\in Q_{s}} p_{q}^{G_s}(y) = O\left (b^{-(i-1)}\sqrt{\frac{b^{i-1}}{m}}\right )=O\left( \frac{1}{\sqrt{mb^{i-1}}}\right ). \label{eq:prod-prob-i-1}\end{align}

For $j=i$, with the upper bound in Claim \ref{claim:lrw-cycle-distribution} (or Claim \ref{claim:lrw-line-distribution} if $i=k-1$), we have $\max_y p^{(i-1)}_q(y) = O ( \sqrt{\frac{b^{i}}{m}})$, which, gives, with Eq.~\eqref{eq:prod-prob-i-1}:
\begin{align*} \prod_{s=0}^{i} \max_{y\in C_b, q\in Q_{s}} p_{q}^{G_s}(y) = O\left ( \sqrt{\frac{b^{i}}{m}} \cdot \frac{1}{\sqrt{mb^{i-1}}}\right ) = O\left( \sqrt{\frac{b}{m}}\right). \end{align*}

Finally, for $m\geq b^k$, we use again Claims \ref{claim:lrw-line-distribution} and \ref{claim:lrw-cycle-distribution} to show that \[ \max_{y,q\in Q_s}p_q^{G_s}(y)\leq b^{-1}\max_{q\in Q_s} (1+c''e^{-c'''qb^{-2}}),\] for $s\leq k-2$, and \[\max_{y,q\in Q_{k-1}}p_q^{\Z}(y)=O\left (\max_{q_{k-1}\in Q_{k-1}} \frac{1}{\sqrt{q_{k-1}}}\right )=O \left (\sqrt{\frac{b^{k-1}}{m}} \right ).\] Hence, we have with Lemma \ref{lem:prod-exponentials},
\[ \prod_{s=0}^{k-1} \max_{y, q\in Q_{s}} p_{q}^{G_s}(y) = 
O\left ( b^{-(k-1)} \sqrt{\frac{b^{k-1}}{m}} \right ) 
= O\left( \frac{1}{\sqrt{b^{k-1}m}}\right).\]
This concludes the proof of Lemma~\ref{lem:prod-proba-rw}.
\end{proof}

\subsubsection{Noise in the last coordinate}\label{subsubsec:noise}

Recall that the last coordinate $Z_{k-1}$ of $Z$ verifies $Z_{k-1}(m)=R_{k-1}(m)+N_{k-1}(m)$. Since $R_{k-1}$ is a walk on $\Z$ that moves with probability $p_{k-1}/2=\Theta(b^{k-1})$, we can expect that $\lvert R_{k-1}(m) \rvert \approx \sqrt{\frac{m}{b^{k-1}}}$.

With the following Lemma, we show that, when considering that the variables $M_i(m)$ are close to their mean, we have $N_{k-1}(m)=O( \sqrt{\frac{m}{b^{k}}})$ with at least constant probability, and hence the noise $N_{k-1}(m)$ is of lesser order than $R_{k-1}(m)$, at least with constant probability.

\begin{lemma}\label{lem:noise}
There is a constant $c'>0$ such that, for any $m\geq b^k$, with $I=(-u,u)$ and $u=c'\sqrt{\frac{m}{b^{k}}}$,
\begin{equation}\label{eq:noise}\min_{\forall i, q_i\in Q_i}\Pr(N_{k-1}(m)\in I\mid \mathcal{A}_{k-1,x,\mathbf{q}} )=\Omega(1). \end{equation}
\end{lemma}
\begin{proof}
It is enough to prove that there is a constant $c''>0$ such that, for any $u>0$, and any $\mathbf{q}=(q_0,\dots,q_{k-1})\in Q_0\times \dots\times Q_{k-1}$,
\begin{equation}\label{eq:noise-general}\mathcal{N}_{x,\mathbf{q}}:=\Pr\left (\lvert N_{k-1}(m) \rvert <u\mid \mathcal{A}_{k-1,x,\mathbf{q}} \right )\geq 1-c'' \frac{\sqrt{\frac{m}{b^{k}}}}{u-1}. \end{equation}
Since, by Eq.~\eqref{eq:def-N-k-1}, $N_{k-1}(m)= \lfloor J'_{k-1}(m)b^{-(k-1)} \rfloor$, we have $\lvert N_{k-1}(m)\rvert \leq 1+\lvert J'_{k-1}(m)\rvert b^{-k+1}$. Thus,  defining $u'=(u-1)b^{k-1}$, we have: \[ \mathcal{N}_{x,\mathbf{q}}\geq 
\Pr\left (\lvert J_{k-1}'(m)\rvert <u'\mid \mathcal{A}_{k-1,x,\mathbf{q}} \right )\] By Markov's inequality, we have $\Pr\left (\lvert J_{k-1}'(m)\rvert \geq u'\mid \mathcal{A}_{k-1,x,\mathbf{q}} \right ) \leq \E \left (\lvert J_{k-1}'(m)\rvert \mid \mathcal{A}_{k-1,x,\mathbf{q}} \right ) \cdot \frac{1}{u'} $ and hence:
\[\mathcal{N}_{x,\mathbf{q}}\geq 1-\E\left(\lvert J_{k-1}'(m)\rvert \mid \mathcal{A}_{k-1,x,\mathbf{q}} \right ) \frac{1}{u'}.\]
Since $J_{k-1}'(m) = \sum_{i\leq k-2} b^i S_{i}(m)$, we have $\lvert J_{k-1}'(m) \rvert \leq \sum_{i\leq k-2} b^i\lvert S_{i}(m)\rvert$, therefore: \[\E\left (\lvert J_{k-1}'(m)\rvert \mid \mathcal{A}_{k-1,x,\mathbf{q}} \right )\leq \sum_{i\leq k-2} b^i\E \left (\lvert S_{i}(m)\rvert\mid \mathcal{A}_{k-1,x,\mathbf{q}}\right ).\]
Hence,
\begin{equation} \label{eq:noise-sum-expectations} \mathcal{N}_{x,\mathbf{q}}\geq 1-  \frac{
\sum_{i\leq k-2} b^i\E \left (\lvert S_{i}(m)\rvert\mid \mathcal{A}_{k-1,x,\mathbf{q}}\right )}{u'}.\end{equation}
Our next goal is to bound $\E \left (\lvert S_{i}(m)\rvert\mid \mathcal{A}_{k-1,x,\mathbf{q}}\right )$. Recall that conditioning on $\mathcal{A}_{k-1,x,\mathbf{q}}$, $S_i(m)$ is a lazy (with parameter $\frac 12$) random walk on $\Z$ with $q_i$ (possibly lazy) steps, and we have $Z_i(m)=x_i$ for every $i\leq k-2$. Thus, for every $i\leq k-2$, $S_{i}(m)+N_i(m)\mod b=Z_i(m)=x_i$. Conditioning on the value $y_i\in C_b$ taken by $N_i(m)$, we have $S_{i}(m)=x_i-y_i\mod b$ and are in the setting of the following claim.

\begin{claim}\label{claim:expectation-modulo-b}
Let $S_q$ be a lazy (with parameter $\frac 12$) random walk on $\Z$ at step $q\geq b^2$, and $x\in[b]$. Then there is a constant $c>0$ such that:
	\[\E(\lvert S_q \rvert \mid S_q =x\mod b) \leq c\sqrt q. \]
\end{claim}
The claim essentially says that the conditioning on $S_q=x\mod b$, for any $x\in [0,b-1]$, does not change significantly the distance travelled by the walk up to step $q$. Let us delay the proof of Claim \ref{claim:expectation-modulo-b} and assume it for now. Then, by Claim \ref{claim:expectation-modulo-b}, for any $y_i\in C_b$, \[\E(\lvert S_{i}(m) \rvert\mid \mathcal{A}_{k-1,x,\mathbf{q}}\cap N_i(m)=y_i)\leq c \sqrt{q_i}.\] Hence, $\E(\lvert S_{i}(m) \rvert\mid \mathcal{A}_{k-1,x,\mathbf{q}})\leq c \sqrt{q_i}$, and thus, by Eq.~\eqref{eq:noise-sum-expectations}:
\[\mathcal{N}_{x,\mathbf{q}} \geq 1- \frac{c}{u'}\sum_{i\leq k-2}b^i\sqrt{q_i}.\] 
Since $q_i\in Q_i$, we have $q_i = \Theta( m b^{-i})$. Hence
\[\sum_{i\leq k-2}b^i\sqrt{q_i}  =\Theta(\sqrt{m} \sum_{i\leq k-2}\sqrt{b}^i)= \Theta(\sqrt m \sqrt b^{k-2}).\]
Thus, for some constant $c''>0$, we have
\[\mathcal{N}_{x,\mathbf{q}} \geq 1-\frac{c''}{u'}\sqrt m \sqrt b^{k-2}.\] 
Replacing $u'$ yields Eq.~\eqref{eq:noise-general} and thus establishes Lemma \ref{lem:noise}, assuming Claim \ref{claim:expectation-modulo-b}.\end{proof}

 We next proceed to prove Claim \ref{claim:expectation-modulo-b}.
\begin{claimproof}[Proof of Claim \ref{claim:expectation-modulo-b}]
	Let $x\in \{0,\dots,b-1\}$. By definition,

	\begin{align} \E(\lvert S_q \rvert &\mid S_q=x\mod b) = \frac{1}{\Pr(S_q=x\mod b)} \sum_{k\geq 1} k \Pr(\lvert S_q\rvert =k \cap S_q=x\mod b) \nonumber \\
	&= \frac{1}{\Pr(S_q=x\mod b)} \sum_{k\geq 1}\sum_{l\in\Z} k \Pr(\lvert S_q\rvert =k \cap S_q=x+lb) \nonumber  \\
	&=  \frac{1}{\Pr(S_q=x\mod b)} \sum_{k\geq 1}\sum_{l\in\Z} k( \Pr( S_q=k=x+lb)+ \Pr(S_q=-k=x+lb))  \nonumber \\
	&= \frac{1}{\Pr(S_q=x\mod b)} \left ( \theta_x + \gamma_x \right ). \label{eq:exp-S-q-cond-mod-b} 
	\end{align}
	where $ \theta_x = \sum_{l\geq 0} (x+lb)\Pr(S_q=x+lb)$ and $\gamma_x=\sum_{l\geq 1} (lb-x)\Pr(S_q=-lb+x)$. We will prove that $\gamma_x+\theta_x$ is of order $\frac{\sqrt{q}}{b}$. For this, note that \begin{equation}
	    \sum_{y=0}^{b-1} \theta_y +\gamma_y=\E(\lvert S_q\rvert ) = O(\sqrt{q}).
	\end{equation}
	Next, let us prove that $\theta_y+\gamma_y$ does not significantly depend on $y\in [b]$, for $q\geq b^2$. First, by symmetry of the process, for any $y\in\{0,\dots,b-1\}$, we have $\gamma_y=\sum_{l\geq 1} (lb-y)\Pr(S_q=lb-y)=\sum_{l\geq 0} (lb+b-y)\Pr(S_q=lb+b-y)=\Theta_{b-y}$. Thus, \begin{equation}
	    \label{eq:sum-gamma-theta}
	    \sum_y \theta_y+\gamma_y=2\sum_y \theta_y=O(\sqrt{q})
	\end{equation}
	Furthermore, as $S$ is lazy with parameter $\frac 12$, we have $\Pr(S_q=z)\geq \Pr(S_q=z+1)$ for any $q>0$ and $z\geq 0$. Hence, \begin{equation}\label{eq:theta-ub}
	     \theta_y \leq \sum_{l\geq 0}(y+lb)\Pr(S_q=lb)\leq  \sum_{l\geq 0}(b+lb)\Pr(S_q=lb) \leq b\Pr(S_q=0\mod b) + \theta_0.
	\end{equation} Using the same monotony property of the process, we have \begin{align*} \theta_y &\geq \sum_{l\geq 0} lb \Pr(S_q=(l+1)b) = \sum_{l\geq 0} (l+1)b \Pr(S_q=(l+1)b)-b\sum_{l\geq 0}\Pr(S_q=(l+1)b) \\&\geq \theta_0 - b\Pr(S_q=0\mod b).\end{align*}
	By Claim \ref{claim:lrw-cycle-distribution}, we have, for $q\geq  b^2$, $ \Pr(S_q=0\mod b) =\Theta(\frac 1b)$. Hence \[ \theta_y = \theta_0 \pm \Theta(1) \] and, by summing, we have \[ \sum_y \theta_y = b\theta_0 \pm \Theta(b). \] 
	Since $\sum_y \theta_y =O(\sqrt{q})$, and $q\geq b^2$, this implies $\theta_0=O(\frac{\sqrt{q}}{b})$, and hence, $\theta_y=O(\frac{\sqrt{q}}{b})$. Combined with Eq.~\eqref{eq:exp-S-q-cond-mod-b}, we have, for $q\geq b^2$,
	\begin{align*} \E(\lvert S_q \rvert &\mid S_q=x\mod b) = O\left ( \frac{1}{\Pr(S_q=x\mod b)}\frac{\sqrt{q}}{b} \right)=O\left(b\frac{\sqrt{q}}{b} \right )=O\left ( \sqrt{q}\right ) ,
	\end{align*}
	where in the last equality we use again Claim \ref{claim:lrw-cycle-distribution}. This proves Claim \ref{claim:expectation-modulo-b}.
\end{claimproof}

\subsection{Estimating the number of visits to $0$ and $x$}
Recall, with Lemma \ref{lem:cov-from-pointwise-appendix}, that we want to find $p>0$ and $m_0$ such that \[\frac{\sum_{m=m_0}^{2m_0}\Pr(Z(m)=x)}{\sum_{m=0}^{m_0} \Pr(Z(m)=0)}\] with $m_0 p^{-1}$ as small as possible, since the cover time is then $\tilde{O}(m_0p^{-1})$, by Lemma~\ref{lem:cov-from-pointwise-appendix}. 

Let us explain intuitively how we find the right $m_0$. We want any $x\in [0,n-1]$ to have a reasonable chance to be visited by $Z(m_0)$. As $x=x_0+\dots+x_{k-1}b^{k-1}\leq n$, with nonnegative $x_i$, we have $x_{k-1}\leq \hat{n}$, where we define:
\[\hat{n}:=\lfloor \frac{n}{b^{k-1}} \rfloor. \]
Hence, we are interested in the behaviour of $Z_0,\dots,Z_{k-2},Z_{k-1}$ on $C_b\times \dots \times C_b \times [0,\hat{n}]$. To ensure that any $x\in[0,n-1]$ has a reasonable chance to be visited, we require that every coordinate $R_i$, for $i\leq k-2$, should be mixed. As $R_i$ is a random walk on $C_b$ which moves with probability $p_i/2$, this happens after $\Theta(p_i^{-1}b^2)=O(p_{k-2}^{-1}b^2)$ steps. We also require that the coordinate $R_{k-1}$ has gone to distance at least $\hat{n}$, which needs about $\hat{n}^2 p_{k-1}^{-1}=\Theta(\frac{n^2}{b^{k-1}})$ steps. This leads us to define: \[ m_0:=\max\{p_{k-2}^{-1}b^2,\hat{n}^2p_{k-1}^{-1} \}= c_b^{-1} b^{k-1}\max\{b,\hat{n}^2\}=\Theta\left ( \max\{ b^k,\frac{n^2}{b^{k-1}}\} \right ) ,\]
as the minimal number of steps such that both of these conditions are satisfied. Note that $m_0 \in [c_b^{-1}b^k, c_b^{-1} b^{k+1}]$, with $c_b^{-1}=p_0^{-1}\in (1, 2)$ as is explicit in the definition of the Weierstrassian process.

\subsubsection{Estimating the expected number of visits to $x$}\label{subsubsec:returns-to-x}

\begin{lemma}\label{lem:total-number-visits-x} The expected number of visits to $x$ in between steps $m_0$ and $2m_0$ is:
\begin{equation} \sum_{m=m_0}^{2m_0} \Pr(Z(m)=x) = \Omega\left( \sqrt{\frac{m_0}{b^{k-1}}} \right). \end{equation}
\end{lemma}
\begin{proof}
To lower bound $\Pr(Z(m)=x)$, we use Eq.~\eqref{eq:separation-lb} with $m\in[m_0,2m_0]$, and $I=(-u,u)$ with $u=\frac{m}{b^{k-2}}>1$. Let us write Eq.~\eqref{eq:separation-lb} as the product of four terms: \[\Pr(Z(m)=x)\geq T_1 T_2 T_3 T_4.\] 
\begin{itemize} \item The first term is:
\[T_1 :=\min_{\forall i, q_i\in Q_i}\Pr(N_{k-1}(m)\in I\mid \forall i (Z_i(m)=x_i) \cap (M_i(m)=q_i))=\Omega(1).\]
where the last inequality is by Lemma \ref{lem:noise}.
\item The second term of Eq.~\eqref{eq:separation-lb} is
 \[T_2:= \min_{y\in I, q\in Q_{k-1}} p^{\Z}_q(x_{k-1}-y),\]
 in which, as $q\in Q_{k-1}$, we have $q=\Theta(mp_{k-1})=\Theta(\frac{m}{b^{k-1}})$. As $\lvert x_{k-1} \rvert \leq \hat{n}=\lfloor\frac{n}{b^{k-1}}\rfloor $ and $\lvert y\rvert <u=1+c'\frac{m}{b^{k}}$, we have $\lvert x_{k-1}-y\rvert = O( \frac{n}{b^{k-1}}+\frac{m}{b^{k}})=O(\frac{n}{b^{k-1}})$ where we verify the last equality easily by using the fact that $m\in [m_0,2m_0]$. Thus, $\lvert x_{k-1}-y\rvert =O(\hat{n})$. As in addition, $q=\Theta(\frac{m}{b^{k-1}})=\Omega(\hat{n}^2)$ and $p_q^{\Z}$ is the distribution of a lazy random walk on the line, which is given by Claim \ref{claim:lrw-line-distribution}, we have: \[ T_2 = \Omega\left(\frac{1}{\sqrt{q}}\right )=\Omega\left( \sqrt{\frac{b^{k-1}}{m_0}} \right ). \]
\item The third term of Eq.~\eqref{eq:separation-lb} verifies, by Lemma \ref{lem:prod-proba-rw},
 \[T_3:=\prod_{j=0}^{k-2} \min_{y\in C_b, q\in Q_{j}} p_q^{j}(y) = \Omega\left (b^{-(k-1)} \right ).\]
 \item Finally, the fourth term of Eq.~\eqref{eq:separation-lb} verifies, by Lemma \ref{lem:multinomial-concentration}, \[T_4:= \Pr(\forall j\leq k-1, M_j(m)\in Q_j) = \Theta(1). \]
 Altogether, we obtain: 
 \begin{align*}\Pr ( Z(m)=x ) =\Omega( T_1 T_2 T_3 T_4 ) = \Omega\left (\frac{1}{\sqrt{b^{k-1}m_0}}\right),\end{align*}\end{itemize}
which implies that the total expected number of visits to $x$ between steps $m_0$ and $2m_0$ is
\begin{equation*}\sum_{m=m_0}^{2m_0} \Pr(Z(m)=x) = \Omega\left( \sqrt{\frac{m_0}{b^{k-1}}} \right), \end{equation*}
as claimed by Lemma \ref{lem:total-number-visits-x}.
\end{proof}

\subsubsection{Estimating the expected number of returns to the origin} \label{subsubsec:returns-to-0} 
To apply Lemma \ref{lem:cov-from-pointwise-appendix}, we want to bound the expected number of returns to $0$ up to step $m_0$. Ideally, we would like to match the upper bound, found in  Lemma \ref{lem:total-number-visits-x}, on the expected number of visits to $x$, which is $O(\sqrt{m_0 b^{-(k-1)}})$. The following Lemma shows this is nearly the case, up to a factor of $k\log b$. 
\begin{lemma}\label{lem:total-number-visits-0}The expected number of returns to $0$ up to step $m_0$ is
\begin{equation*}\sum_{m=0}^{m_0}\Pr(Z(m)=0) =O\left ( \sqrt{\frac{m_0}{b^{k-1}}} k\log b\right ). \end{equation*}
\end{lemma}
\begin{proof}
To estimate $\sum_{m=0}^{m_0}\Pr(Z(m)=0)$, the strategy, as presented in the main text, starts with the following decomposition:
\[ \sum_{m=0}^{m_0}\Pr\left(Z(m)=0\right) = 1+\frac{1}{2}+\sum_{i=0}^{k-1} \sum_{m=1+b^i}^{b^{i+1}}\Pr\left(Z(m)=0\right) + \sum_{m=1+b^k}^{m_0}\Pr\left(Z(m)=0\right).\]
The main idea is to use that, for $i\leq k-1$, between the steps $b^i$ and $b^{i+1}$, the coordinates $0$ to $i-2$ are mixed, and that we know short-time probability bounds for the coordinates $i-1$ and $i$.

 Precisely, let $i\in [1,k-1]$ and $m\in (b^i,b^{i+1}]$. Recall that Lemma \ref{lem:separation-ub} states that:
 	\begin{equation}\label{eq:z-decom-ub} \Pr(Z(m)=0) \leq \sum_{j=0}^i \left( \Pr(M_j(m)\notin Q_j)  \prod_{s=0}^{j-1} \max_{q_s \in Q_s} p_q^{G_s}(0) \right )+ \prod_{s=0}^{i} \max_{q_s \in Q_s} p_q^{G_s}(0).   \end{equation}
By Eq.~\eqref{eq:prod-proba-rw-ub} in Lemma \ref{lem:prod-proba-rw} and Eq.~\eqref{eq:M_i-concentration}, we have:
 \begin{align}\label{eq:ub-proba-0} \Pr(Z(m)=0)= O\left ( \sum_{j=0}^{i-1}\left( e^{-cmb^{-j}}b^{-j}\right) + e^{-cmb^{-i}}\frac{1}{\sqrt{m}\cdot \sqrt{b^{i-1}}} + \frac{\sqrt{b}}{m} \right ).
 \end{align}
 Using that, for any $t>0$, $e^{-t}\leq 2 t^{-2} $, we have \[\sum_{j=0}^{i-1}\left( e^{-cmb^{-j}}b^{-j}\right)\leq \sum_{j=0}^{i-1}\left( \frac{2}{c^2m^2b^{-2j}} b^{-j}\right)= O\left(\frac{1}{m^2}\sum_{j=0}^{i-1} b^j\right)=O\left (\frac{b^i}{m^2}\right)=O\left(\frac 1m\right), \] where we used in the last equality that $m\geq b^i$. For the middle term of Eq.~\eqref{eq:z-decom-ub}, we use that $e^{-t}\leq t^{-1}$ for any $t>0$. Hence, $e^{-cmb^{-i}}\frac{1}{\sqrt{m}\cdot \sqrt{b^{i-1}}} =O(\frac{\sqrt{b}}{m})  $. Altogether, we have
 \begin{align}\label{eq:ub-proba-0'} \Pr(Z(m)=0)= O\left(  \frac{\sqrt{b}}{m} \right )\end{align}
 We now sum Eq.~\eqref{eq:ub-proba-0'} for $m$ between $b^i$ and $b^{i+1}$:
  \begin{align}\sum_{m=1+b^i}^{b^{i+1}} \Pr(Z(m)=0) =O\left( \sum_{m=1+b^i}^{b^{i+1}}\frac{\sqrt{b}}{m} \right) =O\left( \int_{b^i}^{b^{i+1}}\frac{\sqrt{b}}{u}du\right) \nonumber \\= O\left (\sqrt{b}\log\left (\frac{b^{i+1}}{b^i} \right )\right )=O\left (\sqrt{b}\log b \right ).\label{eq:sum-returns-i}\end{align}
Summing Eq.~\eqref{eq:sum-returns-i} for $i=1,\dots,k-1$, we have:
  \begin{align}\sum_{m=1+b}^{b^{k}} \Pr(Z(m)=0)=O\left (k\sqrt{b}\log b \right )\label{eq:sum-returns-1-to-k-1}.\end{align}
  For $m\in [2,b]$, by Lemma \ref{lem:separation-ub} and Eq.~\eqref{eq:M_i-concentration} applied with $i=0$, and Claim \ref{claim:lrw-cycle-distribution}, we have $\Pr(Z(m)=0)=O(e^{-cc_bm}+\frac{1}{\sqrt{m}})=O(\frac{1}{\sqrt{m}})$. Thus,
  \begin{equation}\label{eq:sum-returns-0}
      \sum_{m=2}^b \Pr(Z(m)=0) = O( \sum_{m=2}^b m^{-\frac 12} ) = O( \sqrt{b} ).
  \end{equation}
Finally, let us bound the expected number of returns to the origin between steps $b^k$ and $m_0$. We use Eq.~\eqref{eq:separation-ub} (with $i=k-1$), Eq.~\eqref{eq:M_i-concentration} and Lemma \ref{lem:prod-proba-rw} to obtain, for $m\geq b^k$, 
 \begin{align*} \Pr(Z(m)=0)&=O\left( \sum_{j=0}^{k-1}\left( e^{-cmb^{-j}}b^{-j}\right) +\frac{1}{\sqrt{b^{k-1}}\sqrt{m}}\right)=O\left ( \frac{1}{\sqrt{b^{k-1}}\sqrt{m}} \right ), \end{align*}
 where in the last equality, we use again that $e^{-t}\leq t^{-2}$, and $m\geq b^k$. Summing this for $m\in (b^k, m_0]$, we use again a comparison to an integral:
  \begin{align} \label{eq:sum-returns-k} \nonumber \sum_{m=1+b^k}^{m_0} \Pr(Z(m)=0) &= O\left (  \sum_{m=1+b^k}^{m_0} \frac{1}{\sqrt{b^{k-1}}\sqrt{m}} \right) = O \left ( \int_{b^k}^{m_0} \frac{1}{\sqrt{b^{k-1}}\sqrt{u}}du \right )\\&= O\left ( \sqrt{\frac{m_0}{b^{k-1}}} \right ).\end{align}
Combining Eqs.~\eqref{eq:sum-returns-1-to-k-1}, \eqref{eq:sum-returns-0} and \eqref{eq:sum-returns-k}, we have:
 \begin{equation*}\sum_{m=0}^{m_0}\Pr(Z(m)=0) = O\left (k\sqrt{b}\log b + \sqrt{\frac{m_0}{b^{k-1}}}\right )=O\left ( \sqrt{\frac{m_0}{b^{k-1}}} k\log b\right ), \end{equation*}
 where we used in that last inequality that $m_0\geq b^k$ and hence $\sqrt{\frac{m_0}{b^{k-1}}}\geq \sqrt{b}$. This concludes the proof of Lemma \ref{lem:total-number-visits-0}.
\end{proof}

 \subsection{Concluding the Proof of Theorem \ref{thm:main-up}}
 Now we have by Lemmas~\ref{lem:total-number-visits-x} and \ref{lem:total-number-visits-0}:
 \[ \frac{\sum_{m=m_0}^{2m_0}\Pr(Z(m)=x)}{\sum_{m=0}^{m_0} \Pr(Z(m)=0)} = \Omega\left( \sqrt{\frac{m_0}{b^{k-1}}} \cdot \frac{1}{\sqrt{\frac{m_0}{b^{k-1}}}k\log b} \right ) = \Omega\left ( \frac{1}{k\log b} \right), \]
 and, by Lemma \ref{lem:cov-from-pointwise}, the cover time of the Weierstrassian random walk with parameter $b$ on $C_n$ is:
 \[ O\left ( m_0 \cdot k\log b \cdot k \log n \right ) = O\left ( m_0 k^2\log b \log n\right ) .\]
  Since we have defined \[m_0=\Theta\left (b^{k-1}\max\{b,\hat{n}^2\} \right)=\Theta\left (b^{k-1}\max\{b,\frac{n^2}{b^{2(k-1)}} \}\right )=\Theta\left (n\max\{\frac{b^k}{n},\frac{n}{b^{k-1}}\}\right) ,\] this concludes the proof of Theorem \ref{thm:main-up}.
 
\qed \end{document}